\documentclass[journal,letterpaper,twocolumn,final,10pt]{IEEEtran}

\usepackage{amsthm}
\usepackage{amsmath}
\usepackage{amsfonts}
\usepackage{dsfont}
\usepackage{mathrsfs}
\usepackage{pifont}
\usepackage{amssymb}
\usepackage{verbatim}
\usepackage{upgreek}
\usepackage{color}
\usepackage[final]{graphicx}
\usepackage{epstopdf}
\usepackage{algorithmic}
\usepackage{algorithm}
\usepackage{epsfig}
\usepackage{eucal}
\usepackage{cite}
\usepackage{subfigure}
\usepackage{balance}
\usepackage{setspace}
\usepackage{siunitx}
\usepackage{gensymb}

\makeatletter
\def\maketag@@@#1{\hbox{\m@th\normalfont\normalsize#1}}
\makeatother

\newcommand{\bGamma}{\mbox{\boldmath{$\Gamma$}}}

\newcommand{\bbm}{\begin{bmatrix}}
\newcommand{\ebm}{\end{bmatrix}}

\newcommand{\bit}{\begin{itemize}}
\newcommand{\eit}{\end{itemize}}

\newcommand{\ben}{\begin{enumerate}}
\newcommand{\een}{\end{enumerate}}

\newcommand{\bdesc}{\begin{description}}
\newcommand{\edesc}{\end{description}}

\newcommand{\bea}{\begin{array}}
\newcommand{\eea}{\end{array}}

\newcommand{\tr}{\mbox{\rm Tr}\, }

\newcommand{\beqa}{\begin{eqnarray}}
\newcommand{\eeqa}{\end{eqnarray}}

\newcommand{\ds}{\displaystyle}

\newcommand{\Comment}[1]{}

\newtheorem{prop}{Proposition}

\def\N{{\mathds N}}

\def\R{{\mathds R}}

\def\C{{\mathds C}}

\def\cC{\mbox{$\CMcal C$}}

\def\cL{\mbox{$\mathcal L$}}
\def\cN{\mbox{$\CMcal N$}}

\newcommand{\be}{\begin{equation}}

\newcommand{\ee}{\end{equation}}

\newcommand{\bzero}{{\mbox{\boldmath $0$}}}

\newcommand{\bm}{{\mbox{\boldmath $m$}}}
\newcommand{\bp}{\mbox{\boldmath $p$}}
\newcommand{\bor}{{\mbox{\boldmath $r$}}}

\newcommand{\bv}{{\mbox{\boldmath $v$}}}
\newcommand{\bx}{{\mbox{\boldmath $x$}}}

\newcommand{\bz}{{\mbox{\boldmath $z$}}}

\newcommand{\bA}{{\mbox{\boldmath $A$}}}

\newcommand{\bI}{{\mbox{\boldmath $I$}}}

\newcommand{\bM}{{\mbox{\boldmath $M$}}}

\newcommand{\bO}{{\mbox{\boldmath $O$}}}

\newcommand{\bR}{{\mbox{\boldmath $R$}}}
\newcommand{\bS}{{\mbox{\boldmath $S$}}}

\newcommand{\bU}{{\mbox{\boldmath $U$}}}
\newcommand{\bV}{{\mbox{\boldmath $V$}}}
\newcommand{\bX}{{\mbox{\boldmath $X$}}}

\newcommand{\bZ}{{\mbox{\boldmath $Z$}}}

\newcommand{\diag}{\mbox{diag}\, }

\newcommand{\btheta}{{\mbox{\boldmath $\theta$}}}

\newcommand{\bnu}{{\mbox{\boldmath $\nu$}}}
\newcommand{\bLambda}{{\mbox{\boldmath $\Lambda$}}}

\newcommand{\bsigma}{{\mbox{\boldmath $\sigma$}}}

\title{Learning Strategies for Radar Clutter Classification}

\vspace{0.1cm}

\author{Pia Addabbo, \IEEEmembership{Senior Member, IEEE}, Sudan Han,\\
Danilo Orlando, \IEEEmembership{Senior Member, IEEE}, and Giuseppe Ricci, \IEEEmembership{Senior Member, IEEE}
\thanks{Pia Addabbo is with Universit\`a degli studi Giustino Fortunato, Benevento, Italy. E-mail: {\tt 
p.addabbo@unifortunato.eu}.} 
\thanks{Sudan Han is with  the  National  Innovation  Institute  of  Defense  Technology, Beijing, China E-mail: xiaoxiaosu0626@163.com.}
\thanks{Danilo Orlando is with the Engineering Faculty of Universit\`a degli Studi ``Niccol\`o Cusano'', 
via Don Carlo Gnocchi 3, 00166 Roma, Italy. E-mail: {\tt danilo.orlando@unicusano.it}.}
\thanks{Giuseppe Ricci is with the Dipartimento di Ingegneria dell'Innovazione,
        Universit\`a del Salento, Via Monteroni, 73100 Lecce, Italy.
        E-mail: {\tt giuseppe.ricci@unisalento.it}.
        }
}

\begin{document}

\maketitle

\begin{abstract}
In this paper, we address the problem of classifying clutter returns into
statistically homogeneous subsets.
The classification procedures are devised assuming latent variables, which represent the classes
to which each range bin belongs, and three different models for the structure of the clutter covariance matrix. 
Then, the expectation-maximization algorithm is exploited in conjunction with cyclic estimation
procedures to come up with suitable estimates of the unknown parameters. Finally, the classification
is performed by maximizing the posterior probability that a range bin belongs to a specific class.
The performance analysis of the proposed classifiers is conducted over synthetic data as well as 
real recorded data and highlights that they represent a viable means to
cluster clutter returns with respect to their range.
\end{abstract}


\section{Introduction}
\label{Sec:Introduction}
In the past ten years, improvements in digital architectures and miniaturization technologies have wielded a 
significant impact in the evolution of radar systems which, consequently, are being equipped with more and more 
reliable and sophisticated functions \cite{ScheerMelvin,richards2013principles}. 
This increase in computational resources has led the radar community
to devise detection/estimation algorithms capable of facing with challenging scenarios and, more importantly,
of capitalizing on specific {\em a priori} knowledge about either the system or the environment or both.
In this context, a few examples related to the structural information about the interference covariance
matrix are provided by \cite{JunLiu03,LiuSun19,FOGLIA2017131,Cai1992,hongbinPersymmetric,CP00,Pascal},
where, at the design stage, it is assumed that the system illuminates the surveillance area through 
a symmetrically spaced linear array of sensors. This assumption lends both the interference covariance 
matrix and the steering vector a special structure which yields interesting processing gains at the price
of an additional computational load \cite{Nitzberg-1980,VanTrees4}.

Other approaches relying on a priori information exploit the possible symmetries in the interference
spectral properties \cite{FOGLIA2017131,fogliaPHE_SS,DeMaioSymmetric}. As a matter of fact,
ground clutter returns collected by a monostatic steady radar experience a symmetric power spectral density
centered around zero-Doppler frequency \cite{Billingsley00,Billingsley01}. Remarkably, such property allows
to double data used to estimate the clutter covariance matrix. 
Therefore, the above knowledge-based strategies represent an effective means to deal with situations where
the amount of training data, used for the estimation of the interference covariance matrix, 
is limited (sample-starved condition) otherwise leading to low-quality estimates and,
consequently, to a detection performance degradation. 
Besides the mentioned approaches, other widely 
used techniques to come up with suitable estimates
of the interference covariance matrix consist in the regularization (or shrinkage) of the sample 
covariance matrix towards a given matrix \cite{WieselHero,Tyler,gerlach}.

However, in practice, it is not seldom to meet situations where the presence of inhomogeneities 
makes the interference properties estimation an even more difficult task due to the fact that
such outliers should be censored as proposed in \cite{629144,767347,851934,922976}.
In these contributions, suitable techniques to detect and suppress the outliers are devised 
in order to make the training set homogeneous.
In fact, the homogeneity assumption for secondary data is a very common in detector 
design \cite[and references therein]{kelly1986adaptive,robey1992cfar,BOR-Morgan,GLRT-based} and
when it is no longer valid the performance degradation might become severe \cite{Melvin-2000}.
A more complete approach to the problem of generating homogeneous training sets would envisage 
an additional architectural layout capable of integrating and fusing
information coming from potential heterogeneous sources to depict a clear picture of the clutter properties. 
These sources can be internal or external to the system and comprise mapping data,
communication links, tracker feedback, or other inputs \cite{KB-melvin,KB-WICKS,KB_RANGASWAMY,KB_BENAVOLI}.

Now, note that environment maps might be useful to identify clutter edges and to cluster data into homogeneous 
subsets, whose cardinality can be increased by exploiting a priori information about the 
clutter properties as described before. 
Thus, classifying (or, otherwise stated, clustering) clutter returns would represent 
a desirable feature for modern radar systems. Examples of clutter classifiers are provided by 
\cite{HaykinClass,HaykinClass2}, where the authors build up a neural network or process suitable features
to distinguish between echoes from weather, birds, and aircrafts. Other classifiers are aimed at identifying
the distribution for clutter data \cite{522628,7436109,5674066}, the specific structure of the clutter covariance 
matrix \cite{VincenzoClass}, or the variability of clutter power over the range bins \cite{8801934}.

In this paper, we focus on the problem of partitioning training data into homogeneous subsets and 
we assume that only partial information about the environment is available at the radar receiver, 
namely that a given number of clutter boundaries is present. Then, we design a classification procedure
capable of partitioning the secondary data set into subsets containing statistically homogeneous data.
To this end, we jointly exploit the expectation-maximization (EM) algorithm \cite{Dempster77}
and the latent variable model \cite{murphy2012machine}. The latter tool allows us to introduce
hidden random variables which represent the classes, namely, uniform clutter regions, to which each 
range cell belongs.
Thus, at the end of the procedure, the clustering is accomplished by estimating 
the {\em a posteriori} probability that a range bin
belongs to a specific class. More importantly, we consider three different models for the covariance matrix of the disturbance and more precisely the following
\begin{itemize}
\item
the disturbance of each class is characterized by its own Hermitian covariance matrix;
\item
different classes share a common structure of the covariance matrix, but they have 
different power values (clutter-dominated environment);
\item
noise returns consist of a thermal noise component (whose power is independent of the class) 
plus a clutter component; as in the previous case clutter returns 
share the same structure of the clutter covariance matrix, but each class is characterized
by its own clutter power.
\end{itemize}
The preliminary performance analysis shows the effectiveness of the proposed methods in clustering data.

The remainder of the paper is organized as follows. The next section contains the problem formulation, whereas
Section III is devoted to the design of the classification architectures. Illustrative examples and discussion
about the classification performance are provided in Section IV. Finally, in Section V, we draw the conclusions 
and lay down possible future research lines. Derivations are confined to the Appendices.

\subsection{Notation}
In the sequel, vectors and matrices are denoted by boldface lower-case and upper-case letters, respectively.
The $(i,j)$th entry of a matrix $\bA$ is indicated by $\bA(i,j)$.
Symbols $\det(\cdot)$, $\tr(\cdot)$, $(\cdot)^T$, and $(\cdot)^\dag$ denote the determinant, trace, transpose, 
and conjugate transpose, respectively.
As to numerical sets, $\N$ is the set of natural numbers, $\R$ is the set of real numbers, $\R^{N\times M}$ is the Euclidean space of $(N\times M)$-dimensional 
real matrices (or vectors if $M=1$), 
$\C$ is the set of 
complex numbers, and $\C^{N\times M}$ is the Euclidean space of $(N\times M)$-dimensional 
complex matrices (or vectors if $M=1$). 
$\bI$ and $\bzero$ stand for the identity matrix and the null vector or matrix of proper size. 
Given $a_1, \ldots, a_N \in\C^{N\times 1}$, $\diag(a_1, \ldots, a_N)\in\C^{N\times N}$ indicates 
the diagonal matrix whose $i$th diagonal element is $a_i$.
The acronym pdf and pmf stand for probability density function and probability mass function, respectively, whereas
the conditional pdf of a random variable $x$ given 
another random variable $y$ is denoted by $f(x|y)$. 
Finally, we write $\bx\sim\cC\cN_N(\bm, \bM)$ if $\bx$ is a 
complex circular $N$-dimensional normal vector with mean $\bm$ and positive definite covariance matrix $\bM$
while given a matrix $\bX=[\bx_1 \cdots \bx_M]\in\C^{N\times M}$, writing $\bX\sim\cC\cN_N(\bm,\bM,\bI)$ means 
that $\bx_i\sim\cC\cN_N(\bm, \bM)$, $i=1,\ldots,M$, and the $\bx_i$s are statistically independent.

\section{Problem Formulation and Preliminary Definitions}
\label{Sec:Problem_Formulation}
Consider a radar system equipped with $N\geq 2$ space, time, or space-time channels which illuminates
the operating area consisting of $K$ range bins. 
The signals backscattered by these range cells are suitably conditioned
and sampled by the signal-processing unit to form $N$-dimensional complex vectors denoted by $\bz_1,\ldots,\bz_K$.
Now, let us assume that, from a statistical point of view, the observed environment is temporally 
stationary, whereas its statistical properties may change over the range due, for instance, 
to the presence of clutter boundaries \cite{Richards}. Otherwise stated, we assume that the set of vectors
can be partitioned into $L$ subsets of statistically homogeneous data; the $l$th subset
is denoted by
\be
\Omega_l = \{ \bz_{i_{l,1}},\ldots,\bz_{i_{l,K_l}} \}
\ee
where $K_l$, $l=1, \ldots, L,$ denotes its cardinality.
Thus, the elements of $\Omega_l$ share the same distributional parameters which are generally different from those
associated to the distribution of $\Omega_m$, $m\neq l$. Specifically, we assume that
\be
[\bz_{i_{l,1}} \cdots \bz_{i_{l,K_l}}] \sim \cC\cN_N(\bzero,\bM_l,\bI), \ l=1,\ldots,L,
\label{eq_zk}
\ee
where $\bM_l$ is unknown.

Summarizing, we are interested in estimating the subsets $\Omega_l$ along with the associated unknown
parameter $\bM_l$, $l=1,\ldots,L$. To this end, in the next section we devise
a classification procedure relying on the joint exploitation of the expectation maximization (EM) 
algorithm \cite{Dempster77} and the latent variable model \cite{murphy2012machine}. 
Moreover, besides the most general structure for the clutter covariance matrix, we consider two additional
models which account for possible clutter power variations and diagonal loading due to thermal 
noise.

\section{Classification Architecture Designs}
\label{Sec:Architecture_Designs}
Data classification task is accomplished by introducing $K$ independent and identically
distributed discrete random variables, $c_k$s say, which take on values in $\{ 1,\ldots, L \}$ with unknown pmf
\be
P(c_k=l) = p_l, \quad k=1,\ldots,K,
\ee
and\footnote{Recall that $\sum\limits_{l=1}^L p_l=1$.} 
such that when $c_k=l$, then $\bz_k\sim\cC\cN_N(\bzero,\bM_l)$. Under this assumption, 
it naturally follows that the pdf of $\bz_k$ can be written as
\begin{align}
f(\bz_k;\btheta) &= \sum_{l=1}^L p_l f(\bz_k|c_k=l;\bM_l) \nonumber
\\
&=E_{c_k}[f(\bz_k|c_k; \btheta))] ,
\end{align}
where $E_{c_k}[\cdot]$ denotes the statistical expectation with respect to $c_k$,
\be
\btheta=\left[ \bp^T, \bsigma^T \right]^T
\ee
$\bp=[p_1 \cdots p_L]^T$, 
$\bsigma= \left[ \bnu^T(\bM_1) \cdots \bnu^T(\bM_L) \right]^T$,
$\bnu(\cdot)$ 
a vector-valued function selecting the generally distinct entries of the matrix argument, and
\be
f(\bz_k|c_k=l;\bM_l)=\frac{1}{\pi^N \det(\bM_l)}\exp\{ -\tr[\bM_l^{-1} \bz_k\bz_k^\dag] \}.
\ee
Now, obtaining possible closed-form maximum likelihood estimates of the unknown parameters, namely
$\bM_1,\ldots,\bM_L$ and $\bp$, is not an easy task 
(at least to the best of authors' knowledge). For this reason, we resort 
to the EM-based algorithms, that provide
closed-form updates for the parameter estimates at each step and reach at least a 
local stationary point. To this end, let us write the joint log-likelihood of $\bZ=[\bz_1 \cdots \bz_K]$
as follows
\begin{align}
\cL(\bZ; \btheta) &= \sum_{k=1}^K \log \sum_{c_k=1}^L f( z_k,c_k; \btheta) \nonumber
\\
&= \sum_{k=1}^K \log \sum_{l=1}^L p_l f( z_k|c_k=l; \bM_l).
\end{align}
As observed before, the EM algorithm is a recursive approach to the estimation of the parameter $\btheta$:
its $h$th iteration is aimed at computing $\hat{\btheta}^{(h)}$ starting from 
the estimate at the previous iteration, $\hat{\btheta}^{(h-1)}$ say,
to form a nondecreasing sequence of log-likelihood values, namely
\be
\cL(\bZ; \hat{\btheta}^{(h)}) \geq \cL(\bZ; \hat{\btheta}^{(h-1)}).
\ee
Obviously, an initial estimate of $\btheta$,  
$\hat{\btheta}^{(0)}$ say, is necessary to initialize the algorithm as well as a reasonable stopping criterion as, for instance, a maximum number of iterations,
$h_{\max}$ say.
The EM consists of two steps referred to as the E-step and the M-step, respectively.
The E-step leads to the computation of the following quantity 
\begin{align}
q_k^{(h-1)}\left(l\right)&=p( c_k=l |  \bz_k; \hat{\btheta}^{(h-1)}) \nonumber
\\
&= \frac{\ds f( \bz_k |  c_k=l; \widehat{\bM}_l^{(h-1)}) \hat{p}_l^{(h-1)} } 
{\ds f( \bz_k ; \hat{\btheta}^{(h-1)}) } \nonumber
\\
&=\ds\frac{\ds f( \bz_k |  c_k=l; \widehat{\bM}_l^{(h-1)}) \hat{p}_l^{(h-1)} } 
{\ds \sum_{l'=1}^L f( \bz_k | c_k=l'; \widehat{\bM}_{l'}^{(h-1)})\hat{p}_{l'}^{(h-1)} },
\end{align}
whereas the M-step requires to solve the following problem 
\begin{align}
&\hat{\btheta}^{(h)}=\arg\max_{\btheta}
\sum_{k=1}^K  \sum_{l=1}^L q_k^{(h-1)}\left(l\right) \log \frac{f( \bz_k| c_k=l; \bM_l)p_l}
{q_k^{(h-1)}\left(l\right)} \nonumber
\\
& \Rightarrow 
\hat{\btheta}^{(h)}=\arg\max_{\btheta}
\left\{
\sum_{k=1}^K\sum_{l=1}^L  q_k^{(h-1)}\left(l\right) \log {f( \bz_k| c_k=l; \bM_l)} \right. \nonumber
\\
& \quad \quad \quad \quad \quad \quad 
\quad \quad \quad
\left.+\sum_{k=1}^K\sum_{l=1}^L q_k^{(h-1)}\left(l\right) \log p_l
\right\}.
\end{align}
Note that the maximization with respect to $p_l$, $l=1,\ldots,L$, is independent of
that over $\bM_l$, $l=1,\ldots,L$, and, hence, we can proceed by separately addressing these two problems.
Starting from the optimization over $\bp$, observe that it can be solved by using 
the method of Lagrange multipliers, to take into account the constraint 
\be
\sum_{l=1}^L p_l=1.
\ee
Thus, it is not difficult to show that
\be
\hat{p}_l^{(h)}= \frac{1}{K} \sum_{k=1}^K q_k^{(h-1)}\left(l\right).
\ee
Finally, in order to come up with the estimates of $\bM_1,\ldots,\bM_L$, we solve the following problem
\be
\widehat{\bsigma}^{(h)}=
\arg\max_{\bsigma}
\sum_{k=1}^K\sum_{l=1}^L  q_k^{(h-1)}(l) \log {f( \bz_k| c_k=l; \bM_l)},
\ee
where three different forms for the $\bM_l$, $l=1,\ldots,L$, are considered, namely
\begin{enumerate}
\item $\bM_l$ is a positive definite Hermitian matrix;
\item $\bM_l=\sigma^2_{c,l}\bM$, where $\sigma^2_{c,l}>0$ represents the clutter power which might vary over the
range profile when a clutter edge occurs, while $\bM$ is the common structure shared by the interference of
the $K$ range bins;
\item $\bM_l=\sigma^2_n\bI + \bR_l$, where $\sigma^2_n>0$ is the unknown thermal noise power and 
$\bR_l\in\C^{N\times N}$ denotes the clutter contribution to the interference of the $l$th range bin whose
rank, $r_l$ say, is assumed for the moment known.
\end{enumerate}
Then, the estimates of the unknown parameters for the above cases are provided by the following propositions.

\begin{prop}\label{Prop:EM-M_estimates01}
Assume that $K\geq N$, then an approximation to the relative maximum point of 
\be
g_1(\bM_1,\ldots,\bM_L)=\sum_{k=1}^K\sum_{l=1}^L  q_k^{(h-1)}\left(l\right) \log {f( \bz_k| c_k=l; \bM_l)}
\ee
has the following expression
\be
\widehat{\bM}^{(h)}_l=\frac{\sum_{k=1}^K q_k^{(h-1)}(l) \bz_k\bz_k^\dag}
{\sum_{k=1}^K q_k^{(h-1)}(l)}, \quad l=1,\ldots,L.
\ee
\end{prop}
\begin{proof}
See Appendix \ref{App:ProofProp1}.
\end{proof}

\begin{prop}\label{Prop:EM-M_estimates02}
Assume that $K\geq N$ and form $2$ for $\bM_l$, then, given the function
\be
g_2(\bsigma_c^2,\bM)=\sum_{k=1}^K\sum_{l=1}^L  q_k^{(h-1)}\left(l\right) \log {f( \bz_k| c_k=l; \sigma^2_{c,l}\bM)}
\ee
where $\bsigma^2_c=[\sigma^2_{c,1} \cdots \sigma^2_{c,L}]^T$, 
an approximation to the relative maximum point can be achieved by means of the following
cyclic procedure with respect to the iteration index $t$, $t=1, \ldots, t_{\max}$,
(with $t_{\max}$ a proper design parameter)
\be
(\hat{\sigma}^2_{c,l})^{(1),(h)} =
\frac{\sum_{k=1}^K  q_k^{(h-1)}(l)  \bz_k^{\dagger} (\bM^{(t_{\max}),(h-1)})^{-1} \bz_k}{N 
\sum_{k=1}^K  q_k^{(h-1)}(l)},
\label{eq1_EM_first}
\ee
\be
\widehat{\bM}^{(t),(h)} =\frac{1}{K}
{\sum_{k=1}^K \sum_{l=1}^L  q_k^{(h-1)}(l) \frac{\bz_k \bz_k^{\dagger} }{(\hat{\sigma}^2_{c,l})^{(t),(h)}}},
\label{eq2_EM}
\ee
$t=1, \ldots, t_{\max}$, and
\be
(\hat{\sigma}^2_{c,l})^{(t),(h)} =
\frac{\sum_{k=1}^K  q_k^{(h-1)}(l)  \bz_k^{\dagger} (\bM^{(t-1),(h)})^{-1} \bz_k}{N 
\sum_{k=1}^K  q_k^{(h-1)}(l)},
\label{eq1_EM}
\ee
$t=2, \ldots, t_{\max}$, $l=1, \ldots, L$.

\end{prop}
\begin{proof}
See Appendix \ref{App:ProofProp2}.
\end{proof}


\begin{prop}\label{Prop:EM-M_estimates03}
Assume that $r_l<N$, $l=1,\ldots,L$, is known and form $3$ for $\bM_l$, then an approximation to the relative maximum point of the function
\begin{multline}
g_3(\sigma_n^2,\bR_1,\ldots,\bR_L)\\
=\sum_{k=1}^K\sum_{l=1}^L  q_k^{(h-1)}\left(l\right) \log {f( \bz_k| c_k=l; \sigma^2_{n}\bI+\bR_l)},
\label{eqn:objectiveFunctionProp3}
\end{multline}
can be obtained as follows
\begin{align}
\hat{\sigma}^{2(h)}_n &= \frac{\ds\sum_{l=1}^L \sum_{n=r_l+1}^N \gamma^{(h-1)}_{l,n}}
{\ds\sum_{l=1}^L \sum_{k=1}^K q^{(h-1)}_k(l) (N-r_l)},
\label{eqn:prop3_sigma_c}
\\
\widehat{\bR}^{(h)}_l(r_l) &=\widehat{\bU}^{(h)}_l \widehat{\bLambda}^{(h)}_l(r_l) (\widehat{\bU}^{(h)}_l)^\dag,
\end{align}
where $\widehat{\bU}^{(h)}_l$ is the unitary matrix whose columns are the eigenvectors 
corresponding to the eigenvalues $\gamma^{(h-1)}_{l,1}\geq \gamma^{(h-1)}_{l,2} \geq \ldots 
\geq \gamma^{(h-1)}_{l,N}$
of the matrix 
\be
\bS_l^{(h-1)}=\sum_{k=1}^K q_k^{(h-1)}(l) \bz_k\bz_k^\dag
\ee
and
\begin{multline}
\widehat{\bLambda}^{(h)}_l=\diag\left(\max\left\{\frac{\gamma^{(h-1)}_{l,1}}{\sum_{k=1}^K q_k^{(h-1)}(l)}
-\hat{\sigma}^{2(h)}_n,0\right\},\ldots,\right.
\\
\left.\max\left\{\frac{\gamma^{(h-1)}_{l,r_l}}{\sum_{k=1}^K q_k^{(h-1)}(l)}-\hat{\sigma}^{2(h)}_n,0\right\},0,\ldots,0\right).
\end{multline}
\end{prop}
\begin{proof}
See Appendix \ref{App:ProofProp3}.
\end{proof}
Note that the last proposition supposes that $r_l$, $l=1,\ldots,L$, is known. However,
it is clear that such assumption does not exhibit a practical value; however, the results
provided by Proposition \ref{Prop:EM-M_estimates03} can suitably be exploited in conjunction with an estimator
of $\bor=[r_1,\ldots,r_L]^T$. To this end, we follow the lead of
\cite{ECCMYan} and exploit the MOS rules to build up the following estimator 
for\footnote{Notice that we are neglecting 
some constants that do not depend on $r_l$ and, hence, do not enter the decision process.} $\bor$
\begin{align}
\hat{\bor} &=\arg\min_{\bor} \left\{ 2\sum_{l=1}^{L}\sum_{m=1}^{r_l}\log
\left( \frac{\gamma^{(h)}_{l,m}}{\sum_{k=1}^K q_k^{(h)}(l)} 
\right)\sum_{k=1}^K q_k^{(h)}(l)
\right. \nonumber
\\
&+2\sum_{l=1}^{L}(N-r_l)\log\left[ (\hat{\sigma}^2_n)^{(h)} \right]
\sum_{k=1}^K q_k^{(h)}(l) \nonumber
\\
&\left. + 2 \sum_{l=1}^{L} r_l \sum_{k=1}^K q_k^{(h)}(l)+\frac{2}{(\hat{\sigma}^2_n)^{(h)}}
\sum_{l=1}^{L}\sum_{m=r_l+1}^N
\gamma^{(h)}_{l,m}
+\xi(\bor)
\right\},
\label{eq25}
\end{align}
where $\xi(\bor)$ is a penalty term related to the number of unknown parameters and has the following expression
$\xi(\bor)=\sum\limits_{l=1}^{L}[r_l(2N-r_l)+1]k_p$ with
\be
k_p=
\begin{cases}
2, & \mbox{AIC},
\\
1+a, \ a\geq 1  & \mbox{GIC},
\\
\log(2KN), & \mbox{BIC}.
\end{cases}
\ee
Once the unknown quantities have been estimated, data 
classification can be accomplished by exploiting the following rule
\be
\forall k=1,\ldots,K: \bz_k\sim\cC\cN_N(\bzero,\widehat{\bM}_{\hat{l}_k})
\ee
where
\be
\hat{l}_k=\arg\max_{l=1,\ldots,L}q_k^{(h_{\max})}(l).
\ee

\section{Illustrative Examples and Discussion}
\label{Sec:Simulation}

In this section, the performance of the three proposed classification architectures are assessed 
drawing upon synthetic data as well as real recorded data. Specifically, in the next section, the analysis
is conducted by means of standard Monte Carlo counting techniques, while in the last section, 
the procedures are applied to the Phase One data.

\subsection{Simulated Data}
In the following, data are generated resorting to independent Monte Carlo trials and using two different models 
for the structure of the clutter covariance matrix. In the first case, we suppose the prevalence of the 
clutter contribution assuming an exponential shaped clutter PSD, whereas, in the second case, we do not neglect 
the thermal noise contribution and model the clutter samples as the summation of the echoes from 
patches at distinct angles. All the numerical examples assume $N=16$, $K=96$, and $L=3$. Moreover, the 
presented analysis consists of a first qualitative part, where the classification outcomes of single Monte
Carlo trial are shown, and a second quantitative part, where the root mean square classification error (RMSCE) is
evaluated over 1000 independent Monte Carlo runs. The classification error is defined as the number of
range bins whose class is not correctly identified.

\subsubsection{Prevalence of the clutter contribution}
\label{Subsec:case1}
The examples considered here are aimed at investigating the behavior of Proposition 1 and 2 when 
\be
\bM_l = \sigma_{c,l}^2\bM_c,
\label{eq_model1}
\ee 
where $\sigma_{c,l}^2$ is the clutter power of the $l$th class, and $\bM_c$ is the common clutter structure, 
such that $\bM_c(i,j)=\rho^{|i-j|}$ with $\rho=0.9$.
It is important to observe that for the considered model, the classification procedure relying on
Proposition 3 cannot be applied due to the fact that $r_l=N$, $l=1,\ldots,L$.

As for the initialization of $p_l$s, we choose equiprobable priors, namely, $p_l=1/L$, whereas the 
initial value of $\bM_c$ is set by generating a random 
Hermitian structure as $\bS = \bX\bX^\dag/\tr(\bX\bX^\dag)$, where $\bX$ is a $N\times K$ 
matrix whose columns are complex Gaussian random vectors with zero mean and identity covariance matrix. 
Finally, the $L$ clutter power levels are initialized as follows:
\begin{enumerate}
	\item for each range bin, compute 
	\be
	g(k) = \frac{1}{N} \bz_k^\dag \bS^{-1}\bz_k, \quad k=1,\ldots,K;
	\ee
	\item sort the above quantities in ascending order, $\tilde{g}(1)\leq\tilde{g}(2)\leq\ldots\leq\tilde{g}(K)$;
	\item the mean values of the $K/L$ subsets of the ordered powers is used to set the initial value of the 
	clutter power levels, namely, 
	\be
	\widehat{\sigma}_{c,l}^2 = \frac{L}{K}\sum_{i=(l-1)\frac{K}{L}+1}^{l\frac{K}{L}}\tilde{g}(i), \quad l=1,\ldots L.
	\ee
\end{enumerate}

As preliminary step, we analyze the requirements of the proposed procedures in terms of 
number of EM iterations. To this end, in Figure \ref{figure1},
we plot the joint log-likelihood of $\bZ$ versus the iteration number for Propositions 1 and 2. 
Specifically, the figure assumes $K_1=24$, $K_2=24$, $K_3=48$, $\sigma_{c,1}^2=20$ dB, $\sigma_{c,2}^2=30$ dB, $\sigma_{c,3}^2=40$ dB, and $t_{max}=10$, where $t_{max}$ is the iteration number for the alternating maximization procedure in Proposition 2.
\begin{figure}[htb] \centering
	\subfigure[]{\includegraphics[width=0.49\columnwidth]{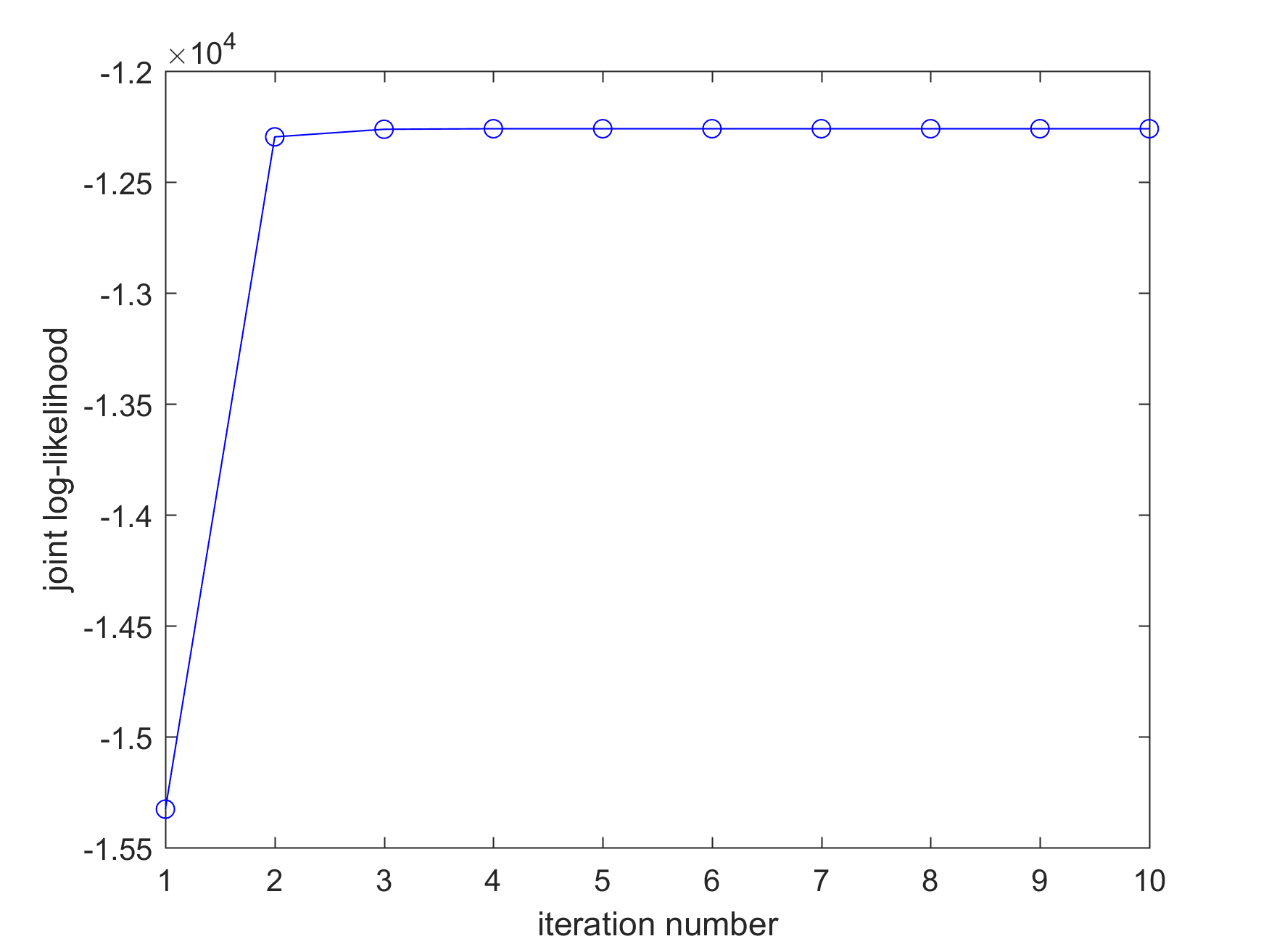}}
	\subfigure[]{\includegraphics[width=0.49\columnwidth]{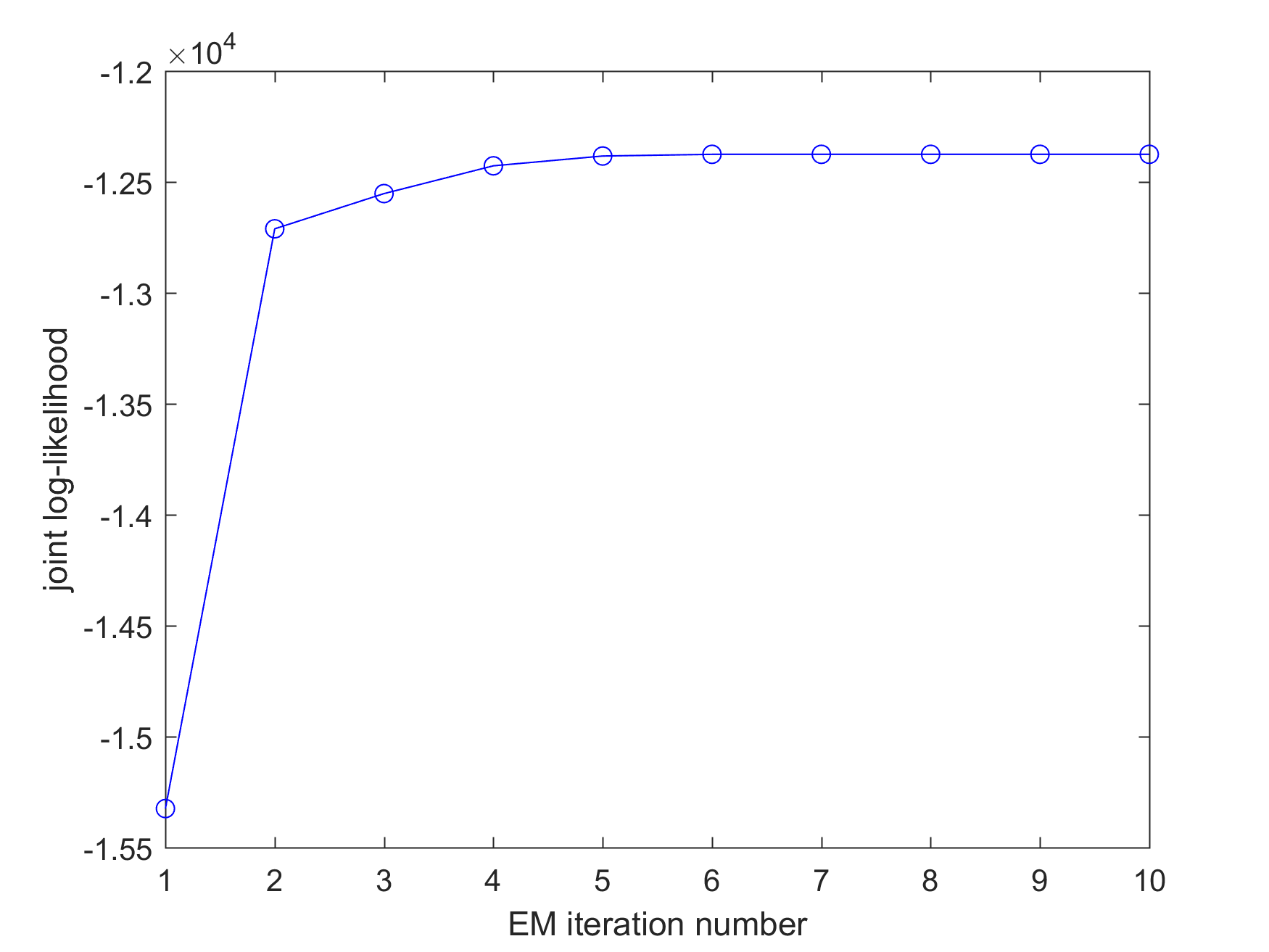}}
	\caption{Joint log-likelihood versus the iteration number of the EM procedure assuming model 1 for the structure
	of the covariance matrix: (a) Proposition 1, (b) Proposition 2.}
	\label{figure1}
\end{figure}
It turns out that, for the considered parameters, 5 iterations are sufficient to achieve convergence. 
Similar results are obtained also for other parameter setting but for brevity are not shown here. 
They point out that 10 iterations are generally sufficient for convergence. 
Therefore, in the next numerical examples, we set $h_{max}=10$. 
As for $t_{max}$, we have also analyzed its effect on the joint log-likelihood and the results 
show that $t_{max}=10$ is a proper choice.

Now, we evaluate the effect of the clutter power levels on the classification performance. 
To this end, we assume $K_1=32$, $K_2=32$, $K_3=32$, and consider 
the following three cases for the clutter power levels: (1) [20,25,30] dB; (2) [20,30,40] dB; 
(3) [20,35,50] dB. Figure \ref{figure2} shows a snapshot (to wit, a Monte Carlo outcome) for the three cases, 
where the estimated clutter classes are represented by "x" red stems, whereas the true classes by the "o" blue stems. 
The results highlight that for the considered parameters and from a qualitative point of view, 
the classification architecture based on Proposition 2 
can achieve better performance than that based on Proposition 1.
\begin{figure}[htb] \centering
	\subfigure[]{\includegraphics[width=0.49\columnwidth]{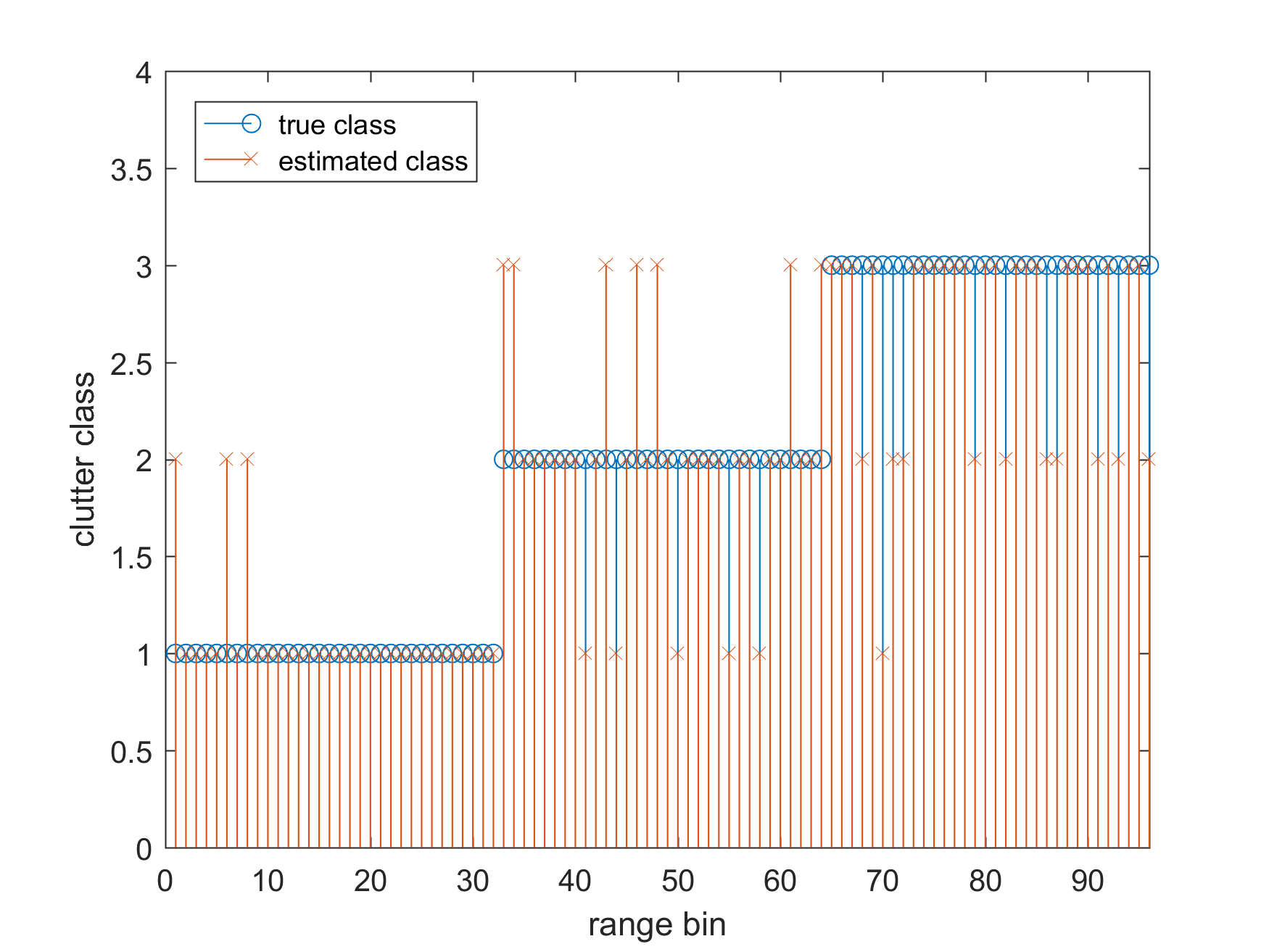}}
	\subfigure[]{\includegraphics[width=0.49\columnwidth]{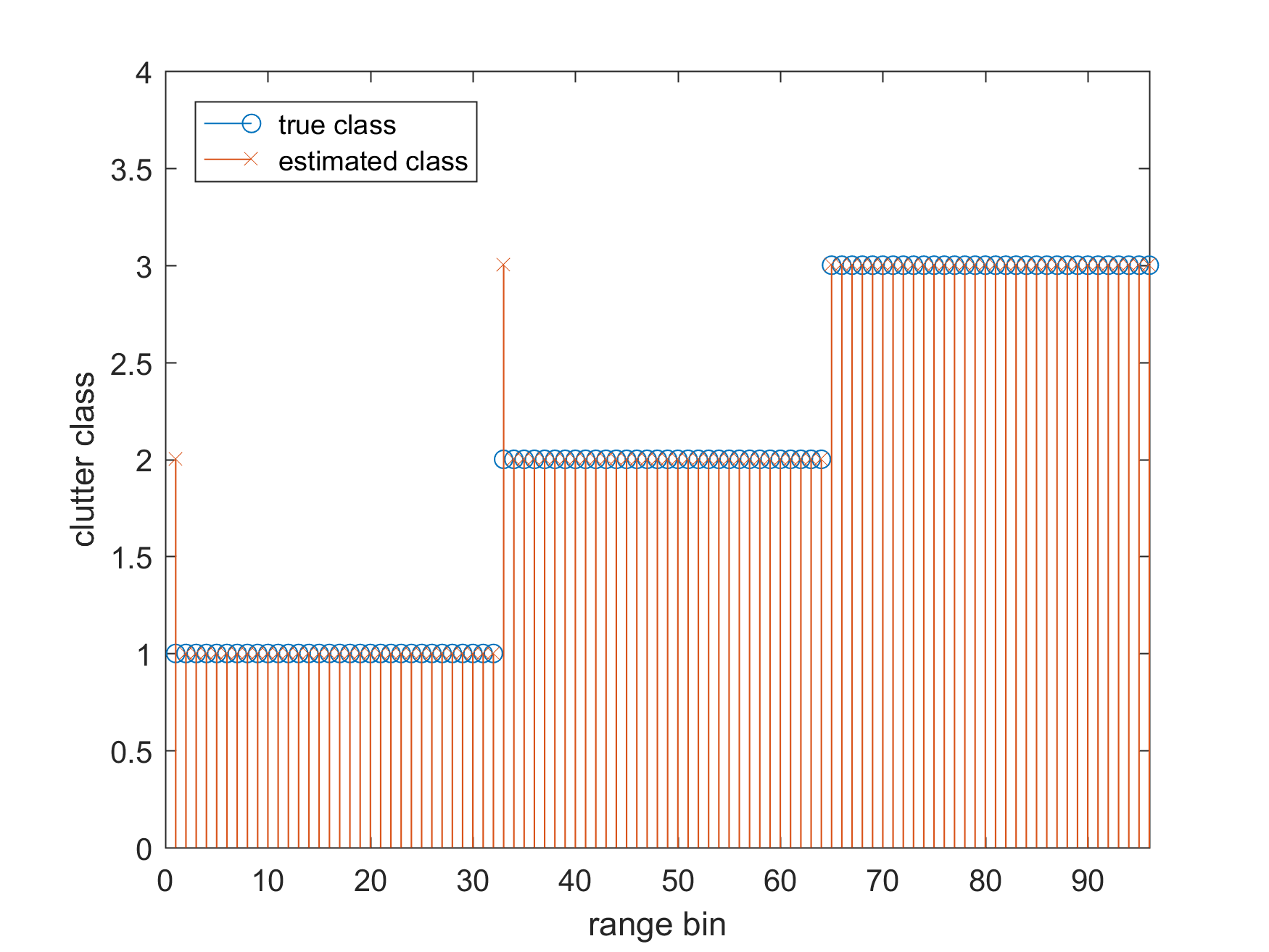}}
	\subfigure[]{\includegraphics[width=0.49\columnwidth]{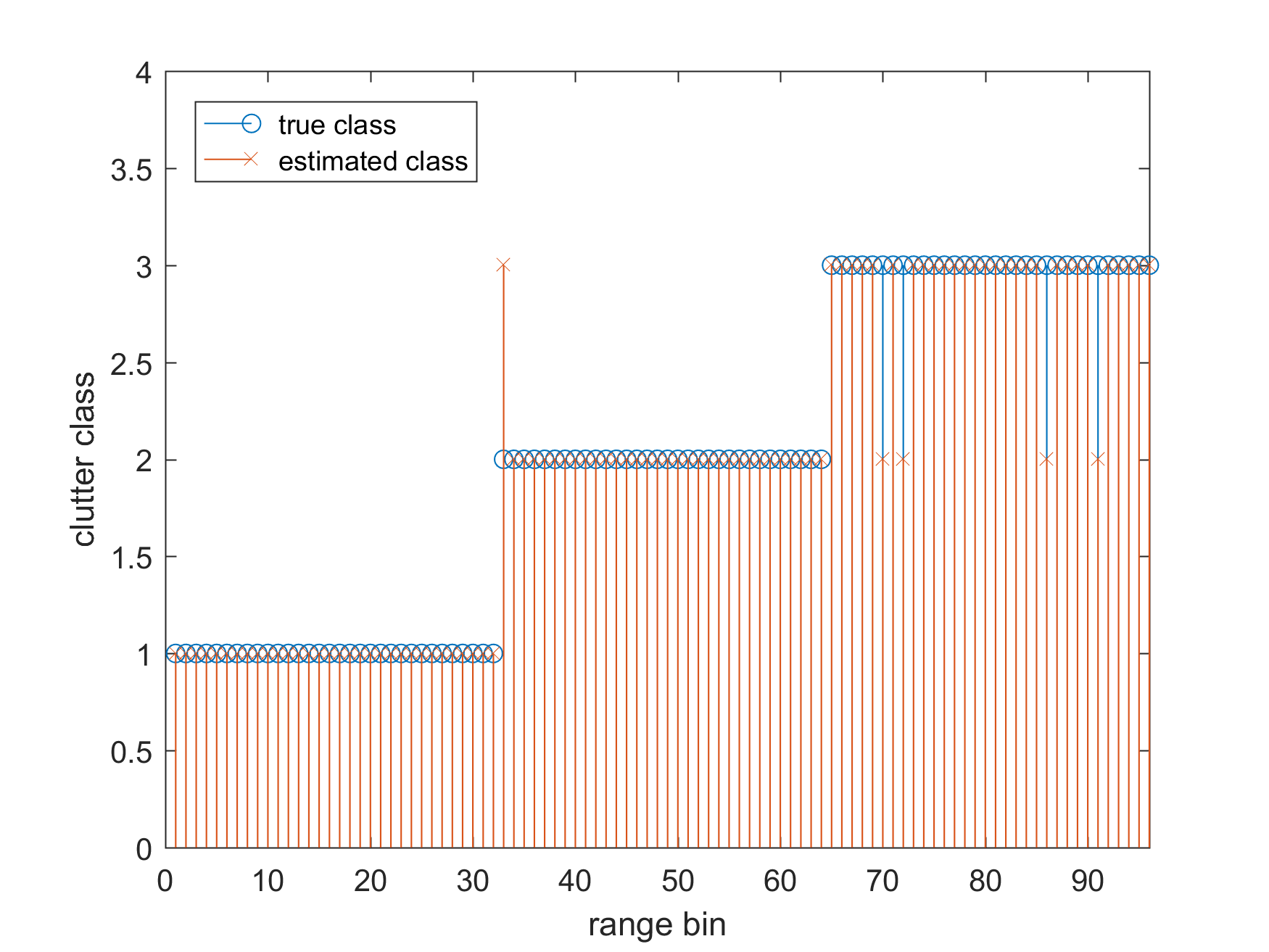}}
	\subfigure[]{\includegraphics[width=0.49\columnwidth]{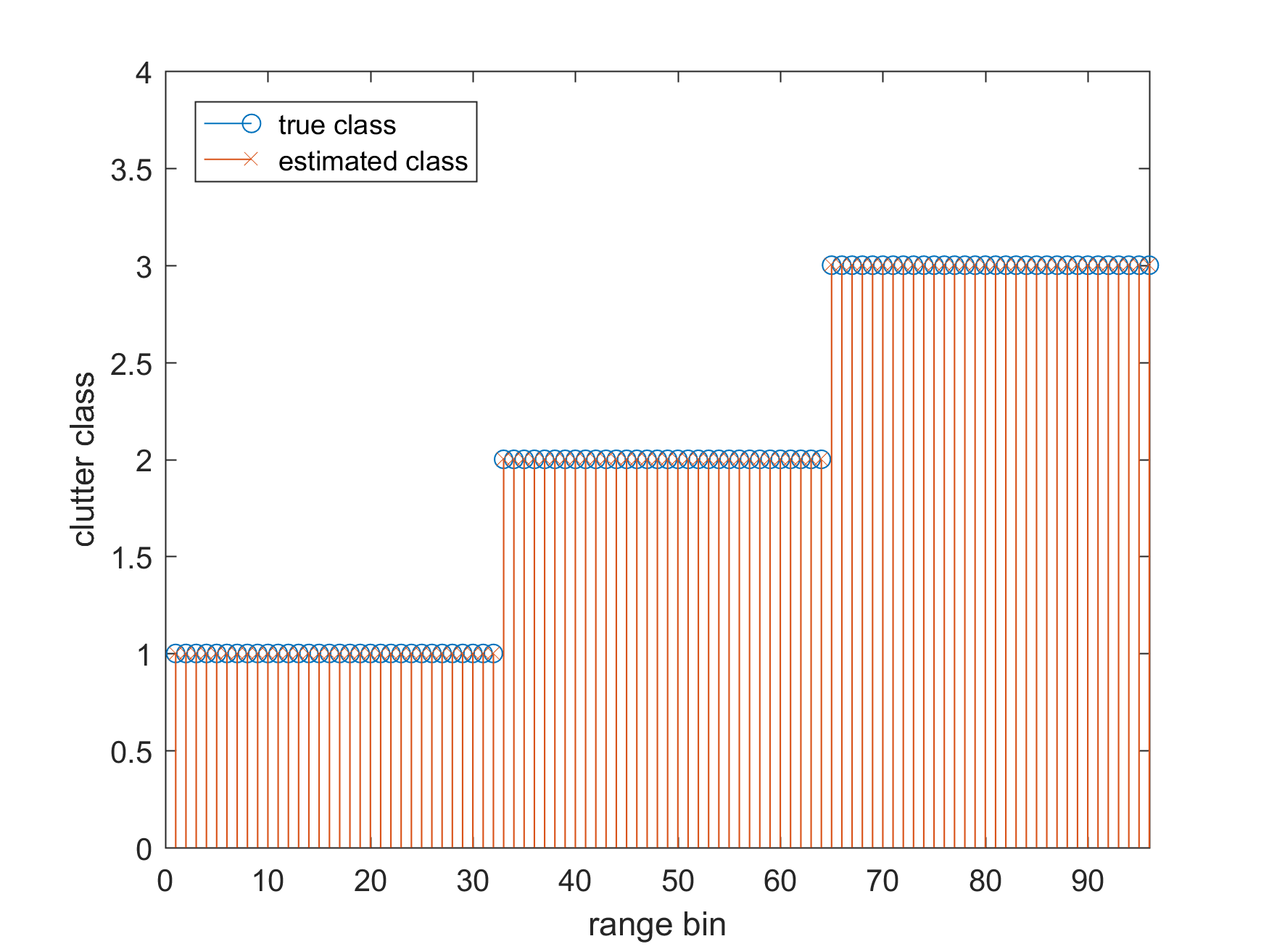}}
	\subfigure[]{\includegraphics[width=0.49\columnwidth]{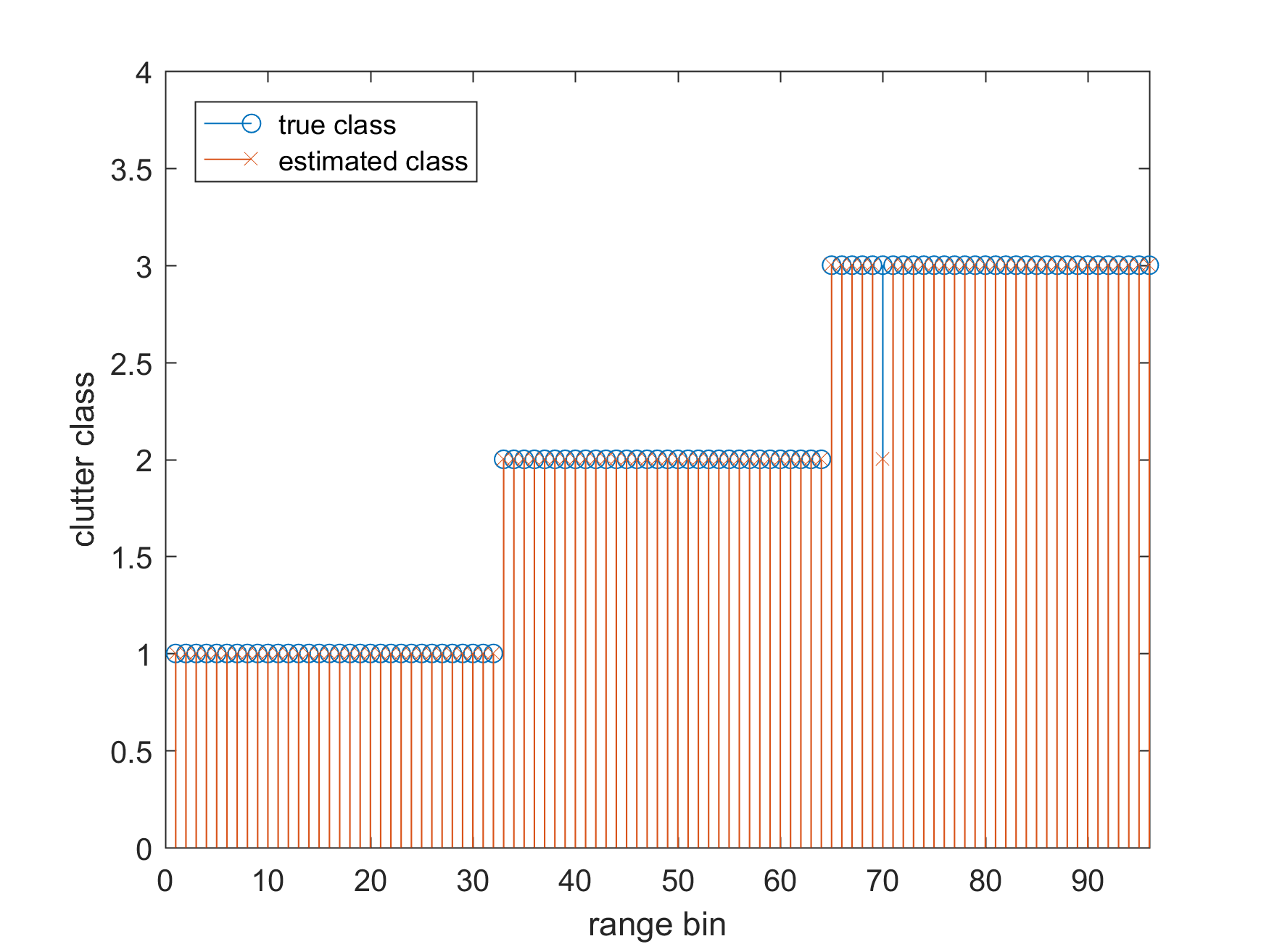}}
	\subfigure[]{\includegraphics[width=0.49\columnwidth]{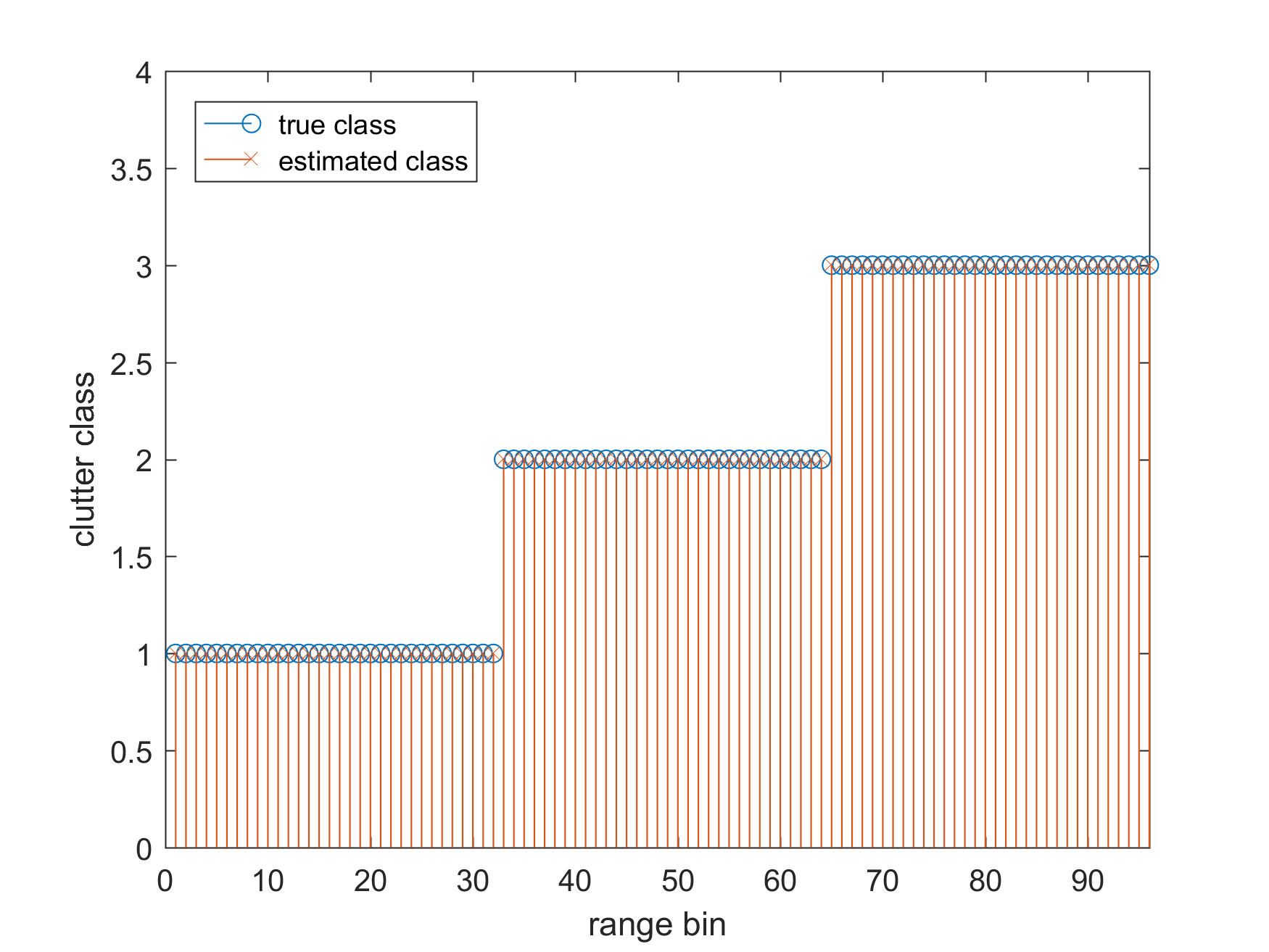}}
	\caption{Classification snapshots for different clutter power levels: (a) case (1) using Proposition 1; (b) case (1) using Proposition 2; (c) case (2) using Proposition 1; (d) case (2) using Proposition 2; (e) case (3) using Proposition 1; (f) case (3) using Proposition 2.}
	\label{figure2}
\end{figure}
A more quantitative analysis can be obtained by resorting to the RMSCE, whose values for the considered scenarios
are reported in Table \ref{table1}.
\begin{table}[htbp]
	\centering
	\caption{RMSCE for covariance model \eqref{eq_model1} and different clutter powers}
	\label{table1}
	\begin{spacing}{1.2}
	\begin{tabular}{|c|c|c|c|}
		\hline
		& case (1) & case (2) & case (3) \\
		\hline
		Proposition 1 & 19.85 & 2.87 & 0.29 \\
		\hline
		Proposition 2 & 3.10 & 0.06 & 0 \\
		\hline
		\end{tabular}
	\end{spacing}
\end{table}
These values confirm the superiority of the algorithm based on Proposition 2 with respect to that relying 
on Proposition 1, indicating that a priori information about the structure of the covariance
matrix can lead to better classification performance. In fact, the simulated covariance matrix structure is more 
compliant with Proposition 2 than Proposition 1. In addition, as expected, the larger the power separation between 
different clutter classes, the lower the error values. 

Finally, we evaluate the effect of different configurations for the $K_l$s on the classification performance 
in terms of the RMSCE assuming $\sigma_{c,l}^2=20+10l$ dB, $l=0,1,2$. 
The classification results are shown in Table \ref{table2}.
\begin{table*}[htbp]
	\centering
	\caption{RMSCE for different values of $K_l$s and covariance model \eqref{eq_model1}}
	\label{table2}
	\begin{spacing}{1.2}{\small
	\begin{tabular}{|c|c|c|c|c|c|c|c|c|c|}
		\hline
		& [20,30,46] & [30,46,20] & [46,20,30] & [24,24,48] & [24,48,24] &[48,24,24] & [18,18,60] & [18,60,18] & [60,18,18]\\
		\hline
		Prop. 1 & 21.14 & 2.52 & 7.25 & 18.63 & 4.67 & 8.04 & 39.18 & 8.50 & 28.33\\
		\hline
		Prop. 2 & 0.13 & 0.09 & 0.09 & 0.10 & 0.09 & 0.08 & 18.99 & 0.09 & 1.04\\
		\hline
	\end{tabular}}
	\end{spacing}
\end{table*}
The superiority of the classification architecture based on Proposition 2 is further validated. Moreover, it is 
worth noticing that the more challenging case for both Propositions is when the number of clutter classes with lower
clutter power is much smaller than that of the clutter class with high clutter power, namely, $K_l=[18,18,60]$. 
This behavior can be explained by the fact that the classification procedures tend to merge small classes with
low powers.

\begin{figure}
    \centering
    \includegraphics[scale=0.4]{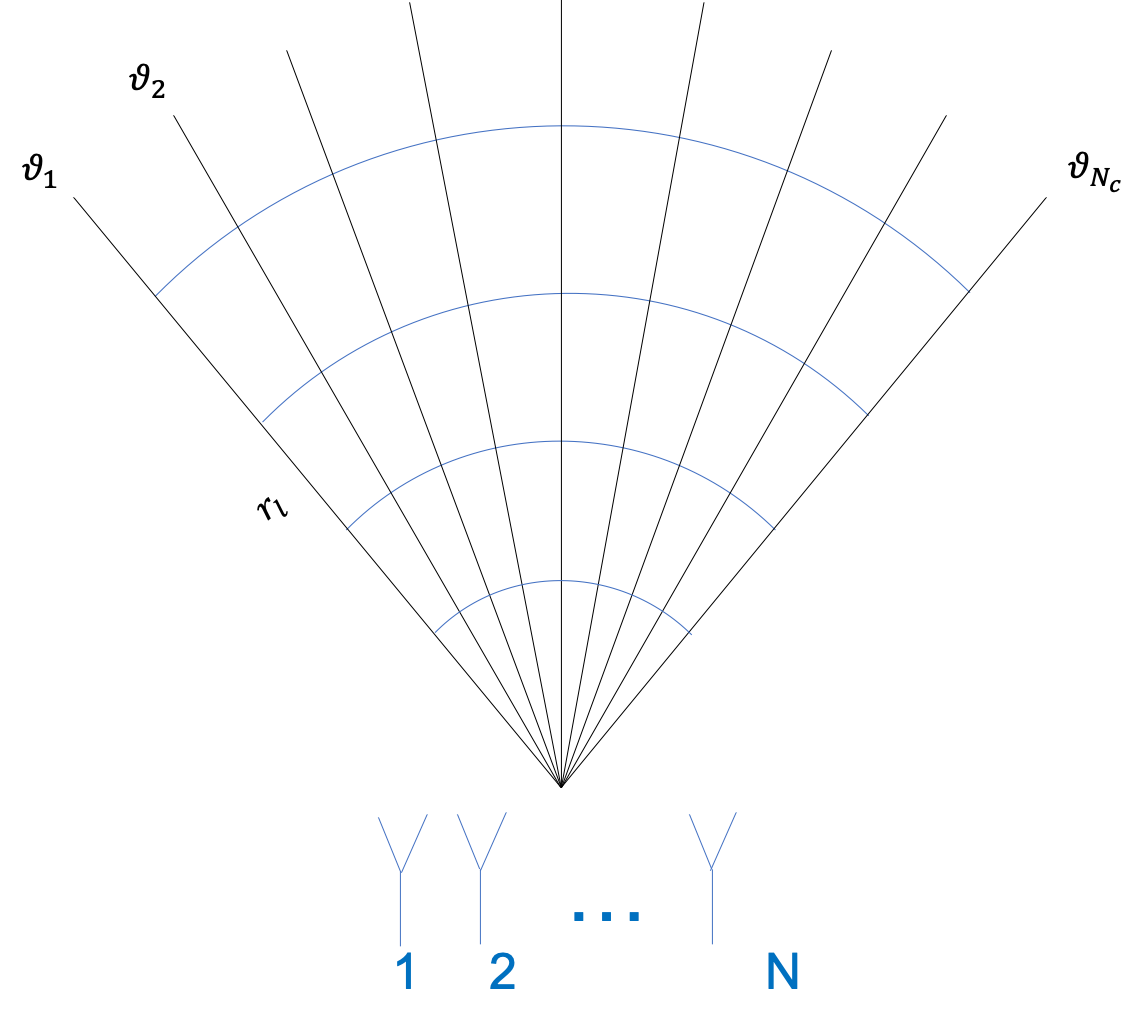}
       \caption{Angular sector under surveillance.}
    \label{fig:angular}
\end{figure}

\subsubsection{Distributed clutter plus thermal noise}
\label{Subsec:case 2}
In this subsection, we assume another clutter covariance matrix model. As shown in Figure \ref{fig:angular}, 
we consider a uniformly-spaced linear array of $N$ identical and isotropic sensors with inter-element 
distance equal to $\lambda/2$, where $\lambda$ is the wavelength corresponding to the radar carrier frequency. 
We only consider the spatial processing for simplicity and model the clutter samples as the summation of 
individual patch returns at 
distinct angles \cite{iet:/content/books/ra/sbra020e}, leading to the following covariance structure
\be
\bM_l = \sigma_{c,l}^2\sum_{\theta_i\in \Theta_l}\bv(\theta_i)\bv(\theta_i)^\dag+\sigma_{n}^2\bI ,
\label{eq_model2}
\ee
where
\begin{itemize}
	\item $\Theta_l=\{\theta_1^l,\theta_2^l,\ldots,\theta_{N_{c}^l}^l\}$ (for simplicity, we suppose that the number of the angular sectors is the same for each class, namely, $N_{c}^l=N_c$ for all $l$);
	\item $\bv(\theta_i)$ is the spatial steering vector whose expression is given by $\bv(\theta_i)
	=\frac{1}{\sqrt{N}}\left[1,e^{j\pi\sin\theta_i},\ldots,e^{j\pi(N-1)\sin\theta_i}\right]^T \in \C^{N\times 1}$.
\end{itemize}

In the following, we set $N_c=5$, the beam pointing direction to $0^\circ$, and an angular sector within 
the first null beamwidth of $14^\circ$, namely, $\Theta_l=\{-5.6^\circ,-2.8^\circ,0^\circ,2.8^\circ,5.6^\circ\}$. 

The initialization method is the same as that in the Subsection \ref{Subsec:case1}. Moreover, as to Proposition 3, 
we consider two situations, i.e., the clutter rank $\bor$ is known and $\bor$ is unknown. In the latter case, 
the GIC rule with $a=2$ is exploited to estimate $\bor$ using \eqref{eq25}.

\begin{figure}[htb] \centering
	\subfigure[]{\includegraphics[width=0.49\columnwidth]{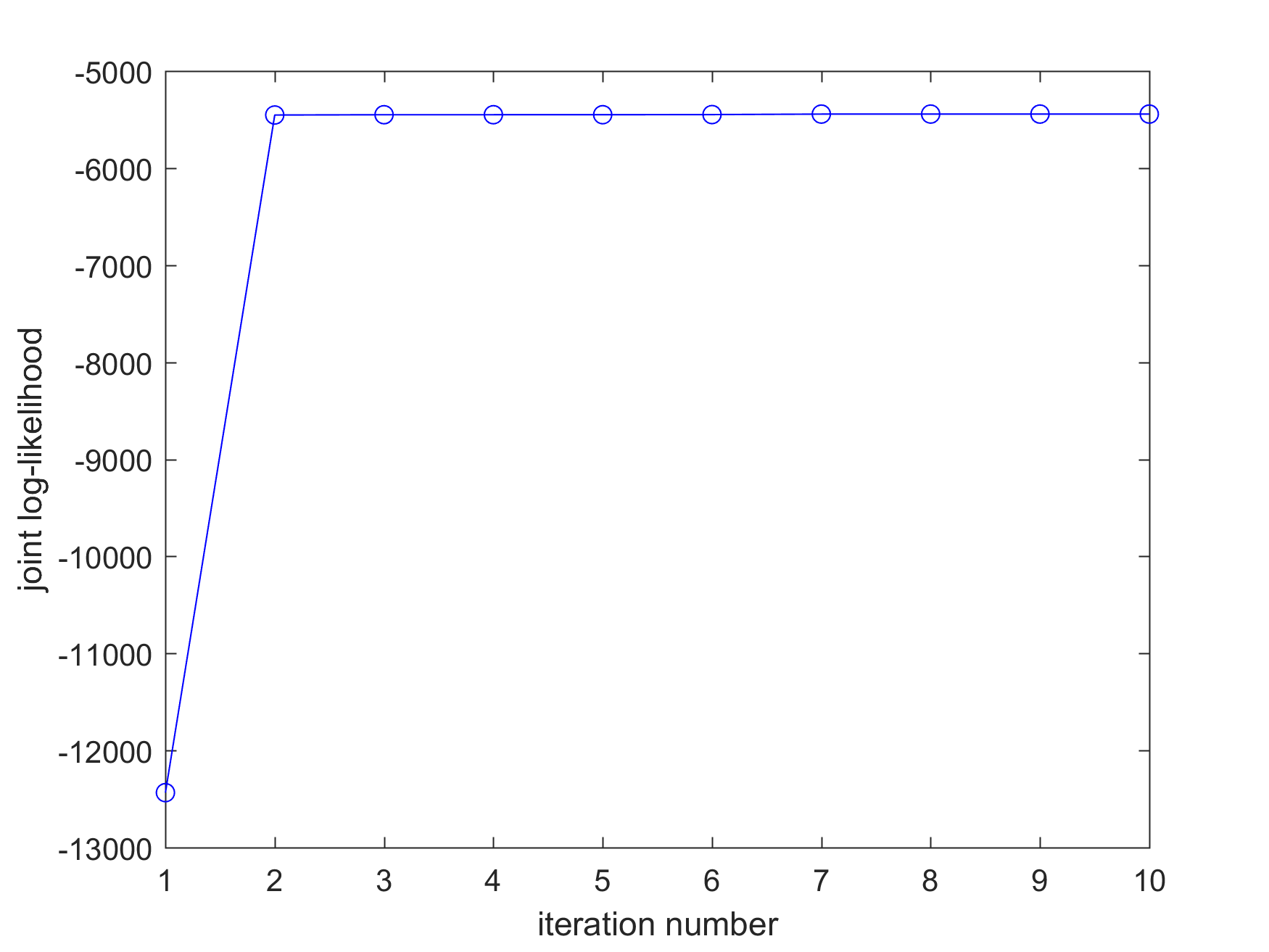}}
	\subfigure[]{\includegraphics[width=0.49\columnwidth]{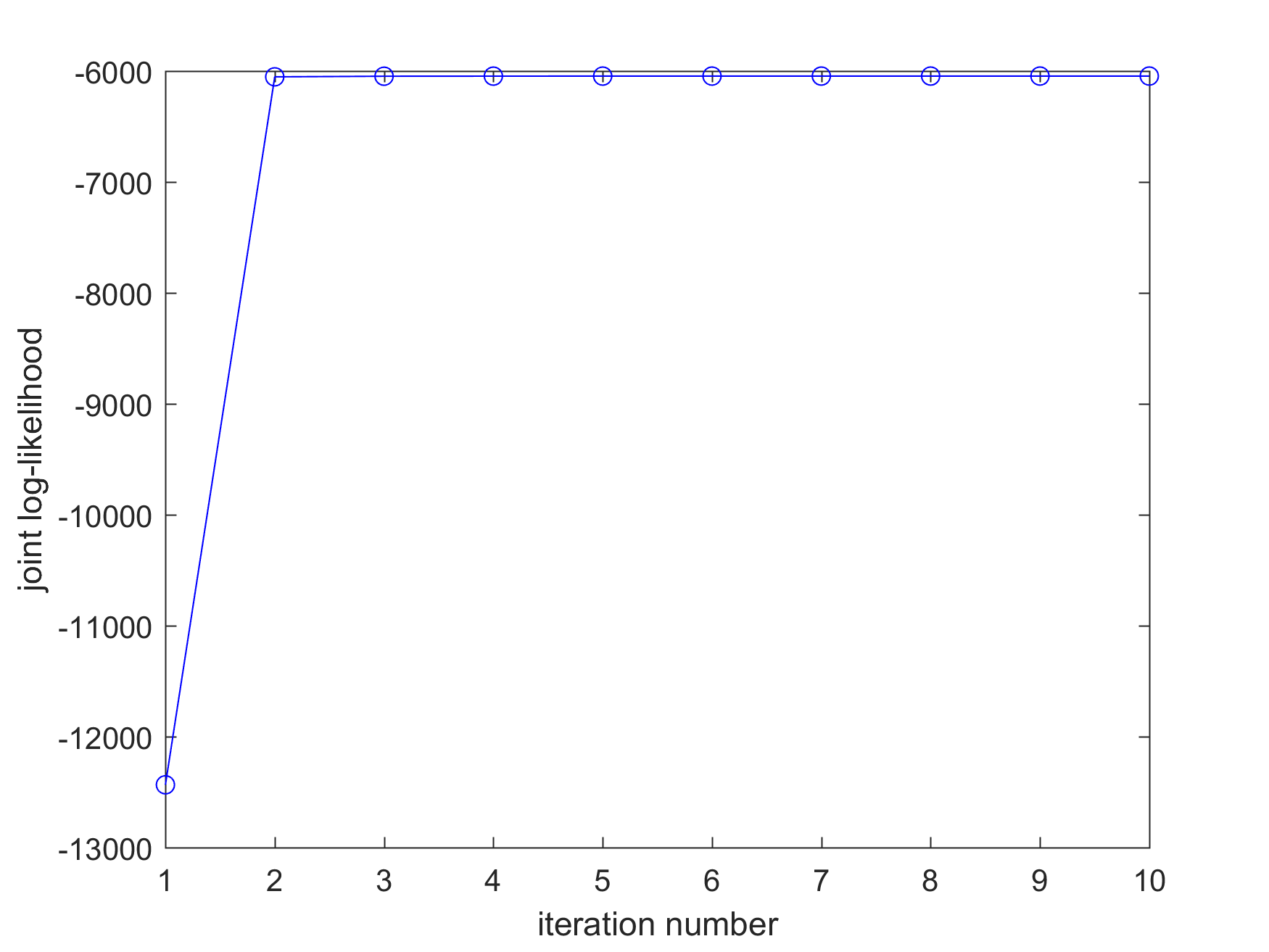}}
	\subfigure[]{\includegraphics[width=0.49\columnwidth]{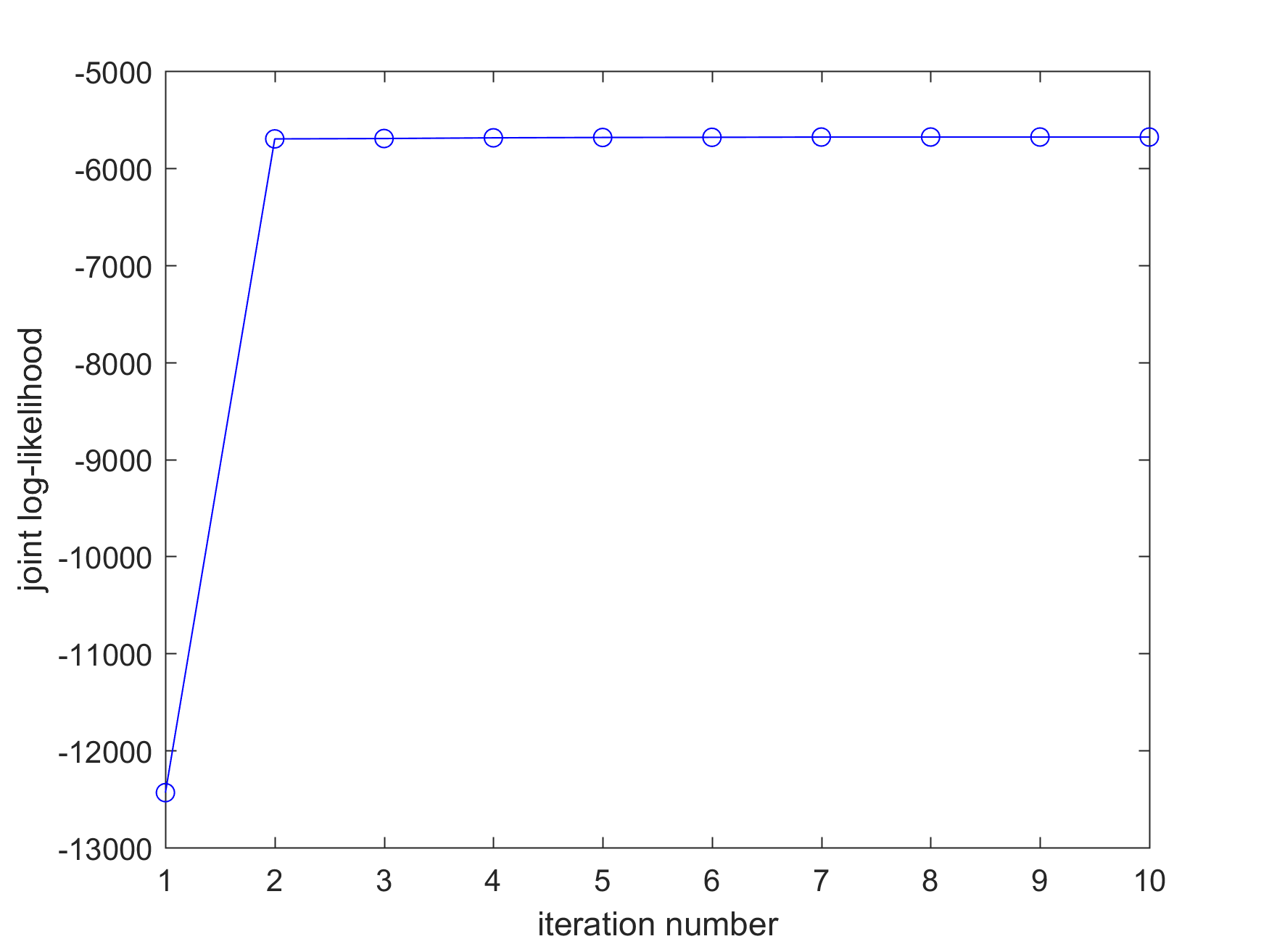}}
	\subfigure[]{\includegraphics[width=0.49\columnwidth]{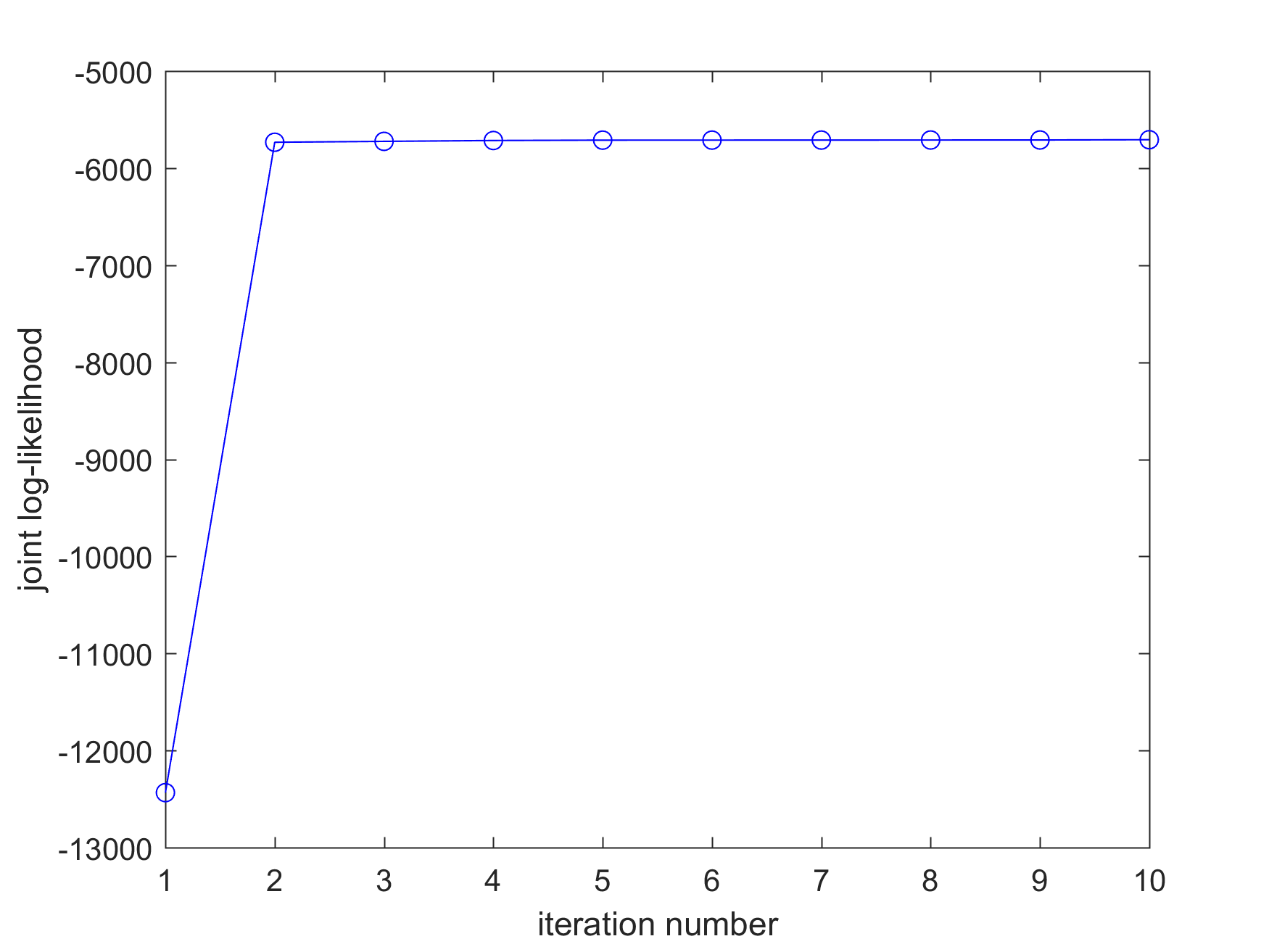}}
	\caption{Joint log-likelihood versus the iteration number of the EM procedure assuming model 2 for
	the structure of the covariance matrix: (a) Proposition 1; (b) Proposition 2; (c) Proposition 3 with 
	known $\bor$; (d) Proposition 3 with unknown $\bor$.}
	\label{figure3}
\end{figure}
Figure \ref{figure3} shows the joint log-likelihood versus the iteration number of the EM procedure for $K_1=32$,
$K_2=32$, $K_3=32$, $\sigma_{c,1}^2=20$ dB, $\sigma_{c,2}^2=30$ dB and $\sigma_{c,3}^2=40$ dB. The curves confirm 
that a few iterations are sufficient for convergence.

\begin{figure}[htb] \centering
	\subfigure[]{\includegraphics[width=0.49\columnwidth]{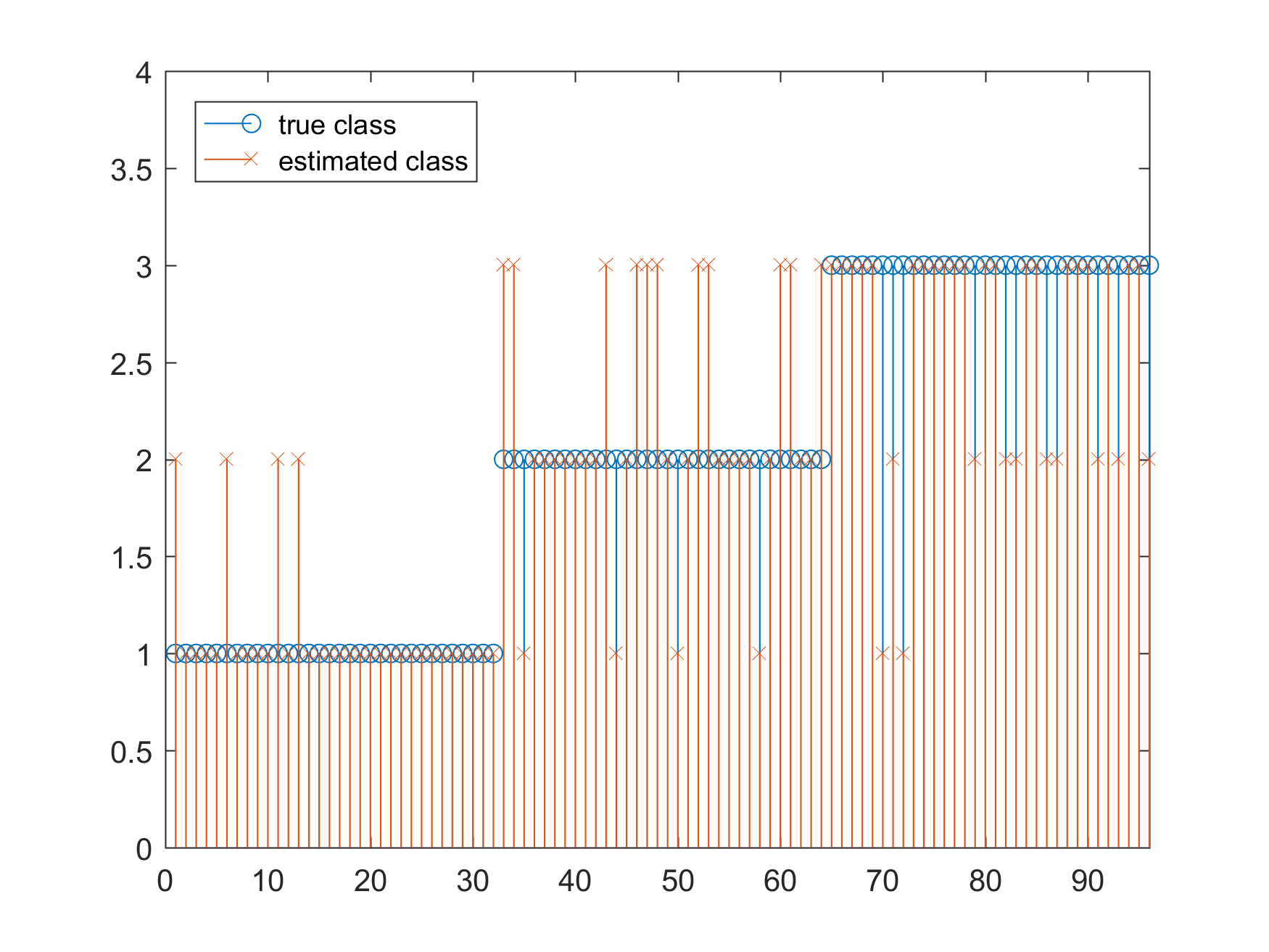}}
	\subfigure[]{\includegraphics[width=0.49\columnwidth]{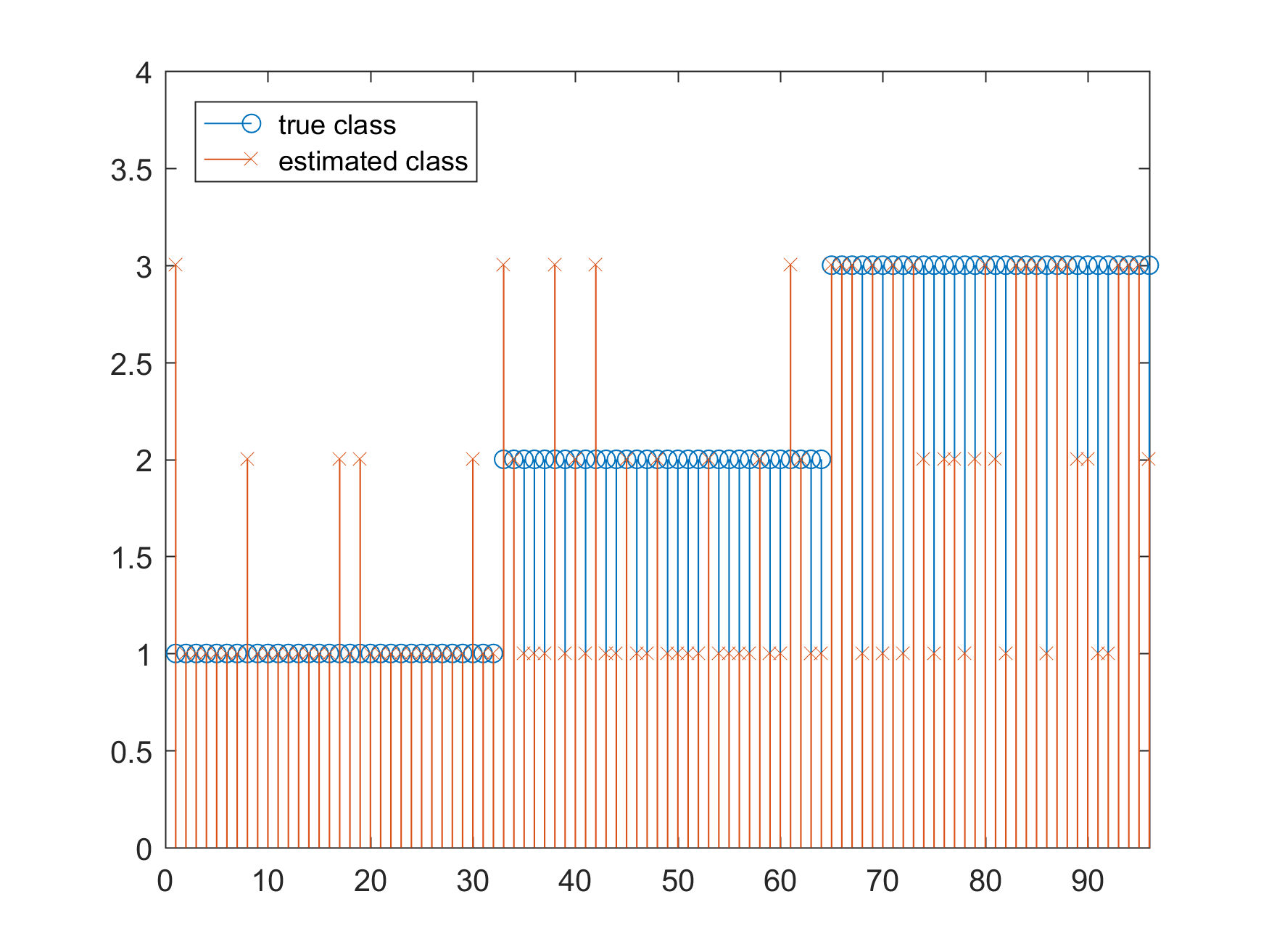}}
	\subfigure[]{\includegraphics[width=0.49\columnwidth]{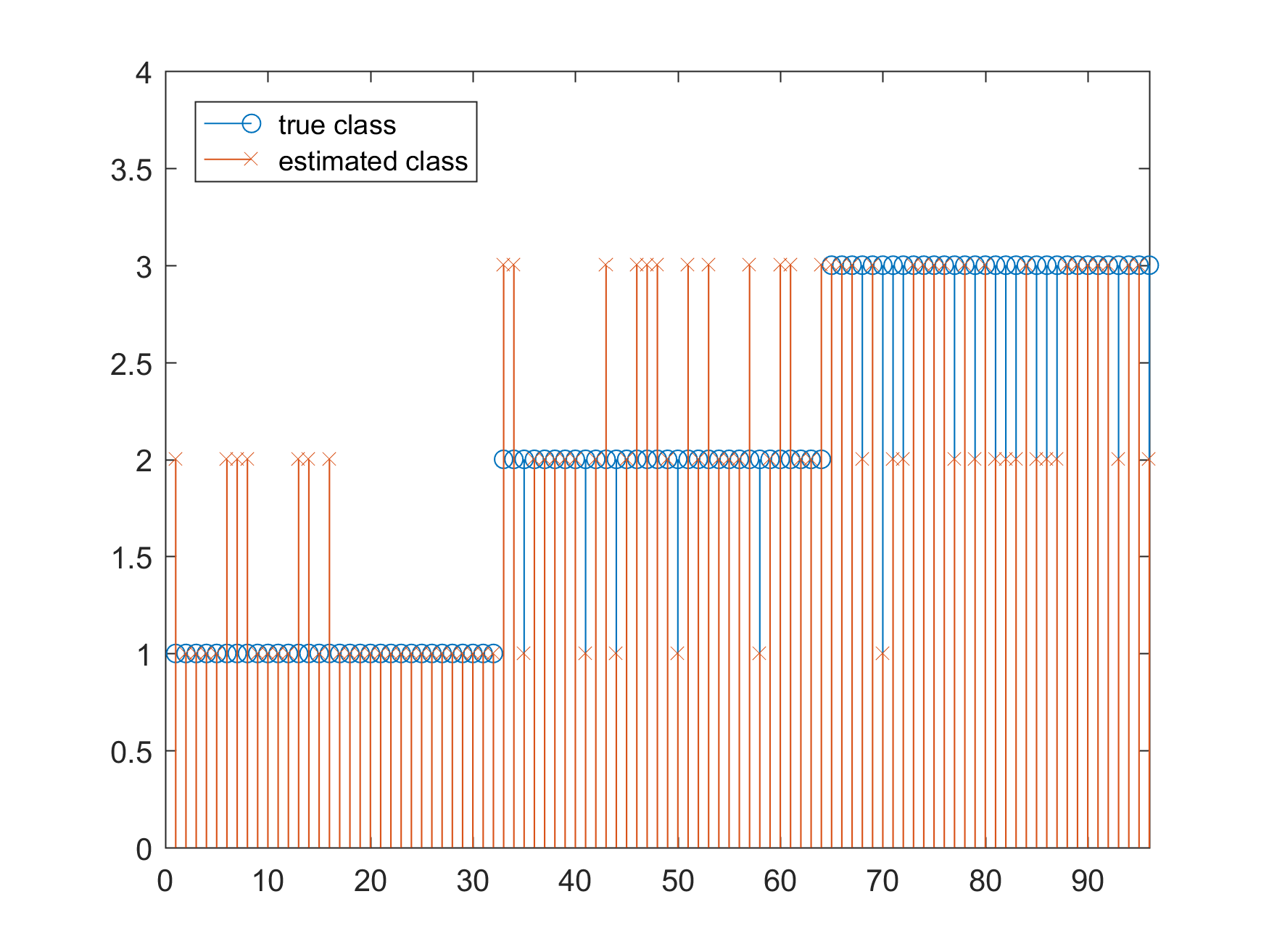}}
	\subfigure[]{\includegraphics[width=0.49\columnwidth]{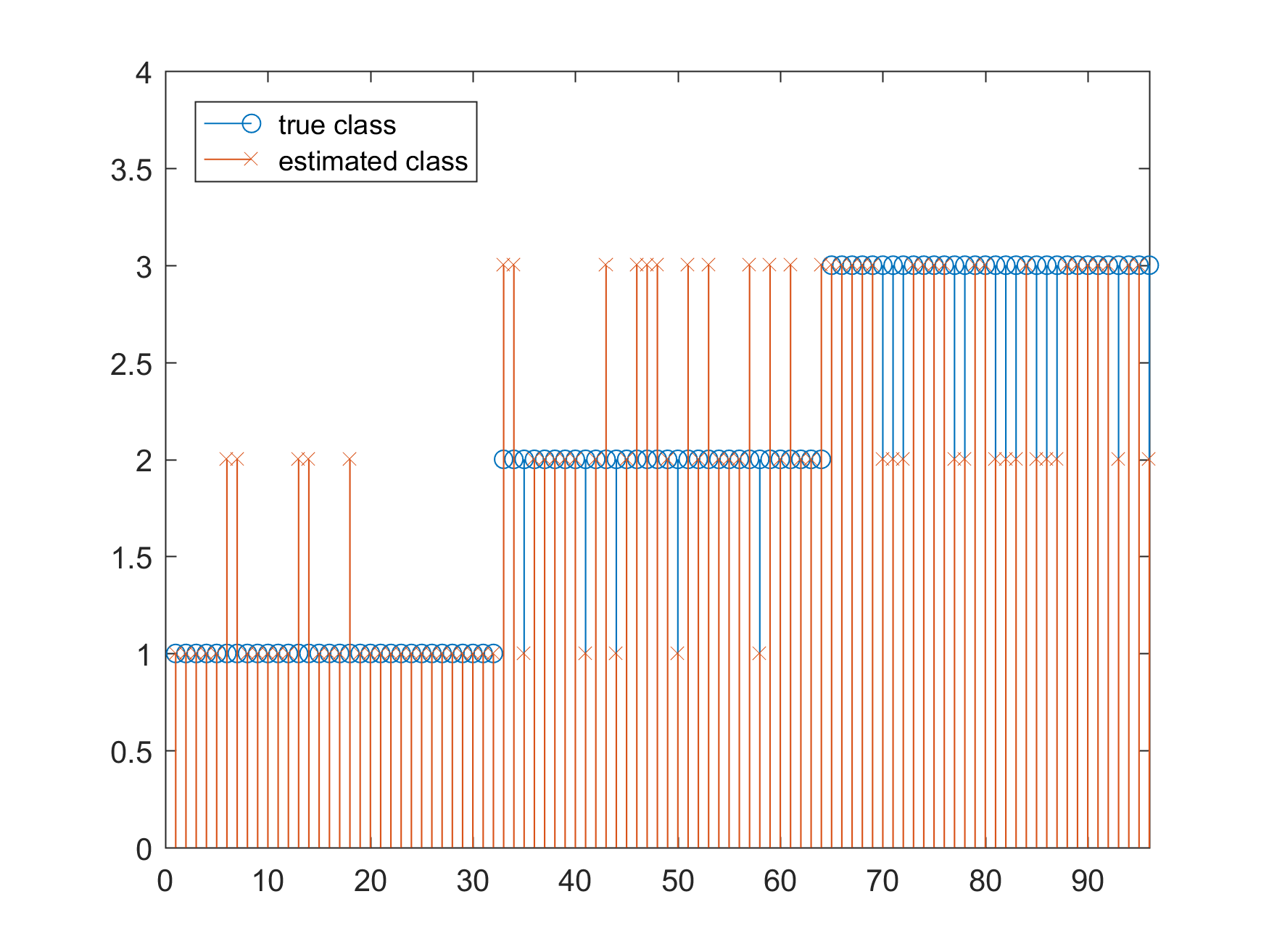}}
	\caption{Classification snapshots for $\sigma_{c,1}^2=20$ dB, $\sigma_{c,2}^2=25$ dB and $\sigma_{c,3}^2=30$ dB: (a) Proposition 1; (b) Proposition 2; (c) Proposition 3 with known $\bor$; (d) Proposition 3 with unknown $\bor$.}
	\label{figure4}
\end{figure}

\begin{figure}[htb] \centering
	\subfigure[]{\includegraphics[width=0.49\columnwidth]{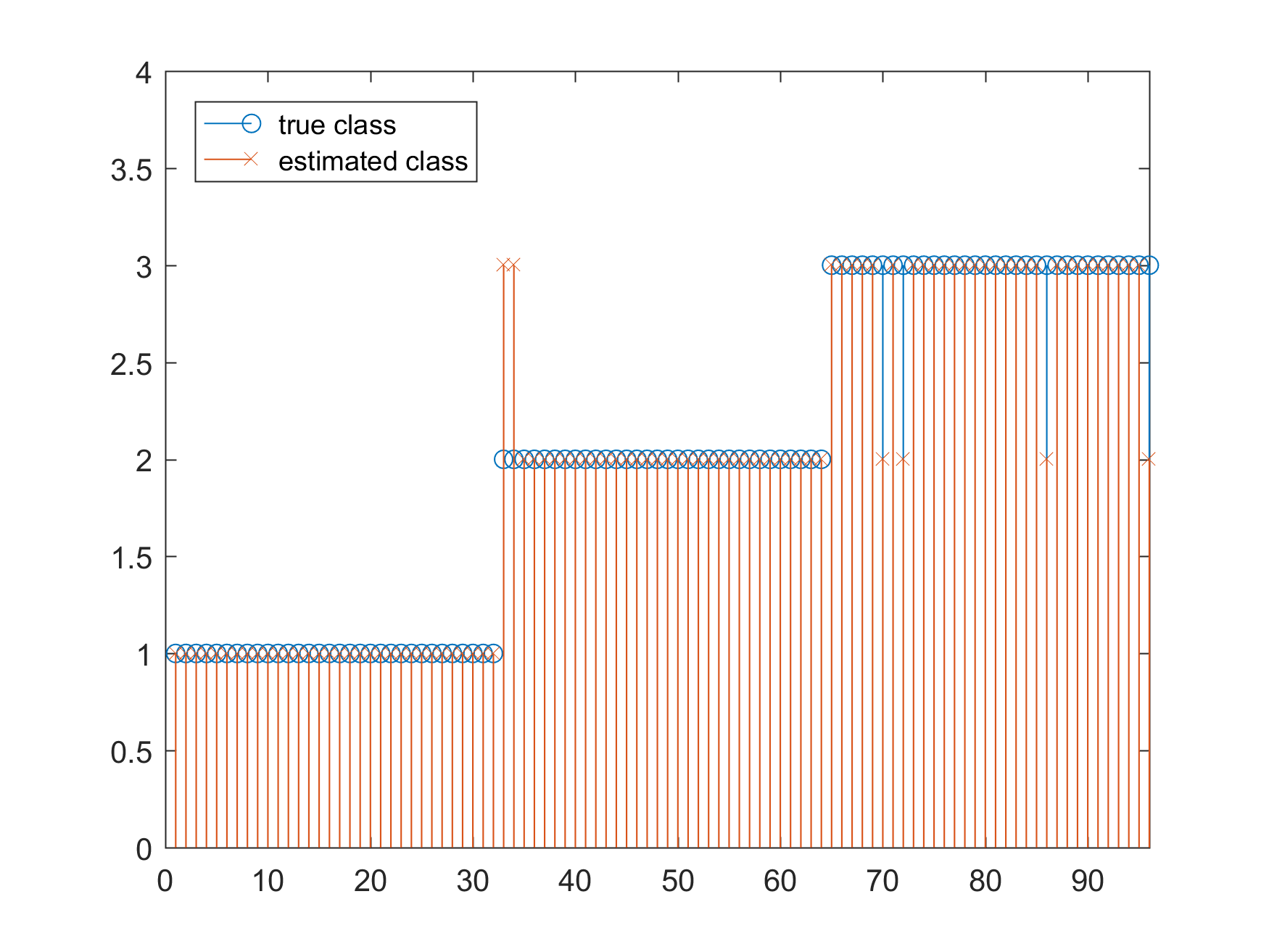}}
	\subfigure[]{\includegraphics[width=0.49\columnwidth]{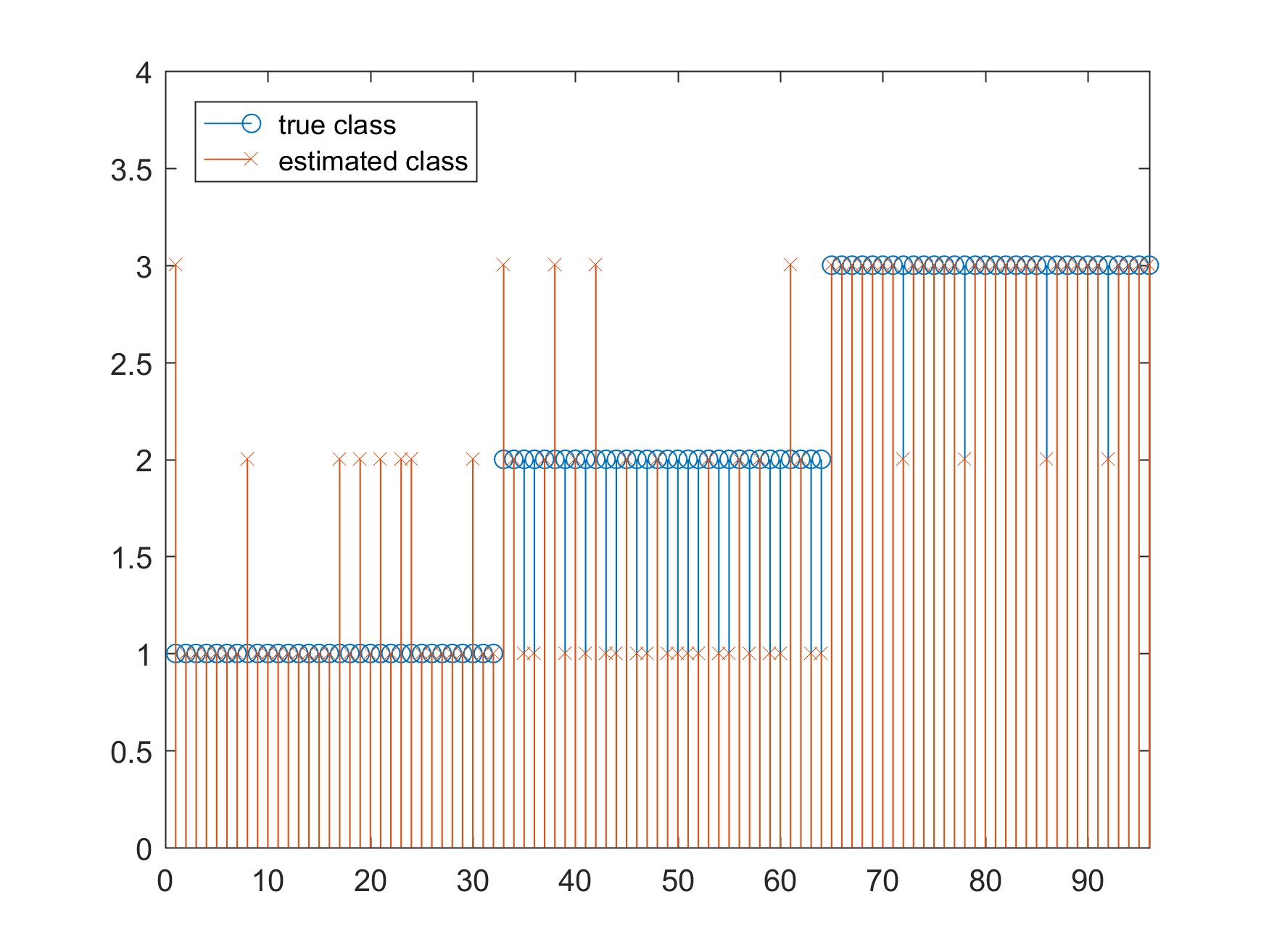}}
	\subfigure[]{\includegraphics[width=0.49\columnwidth]{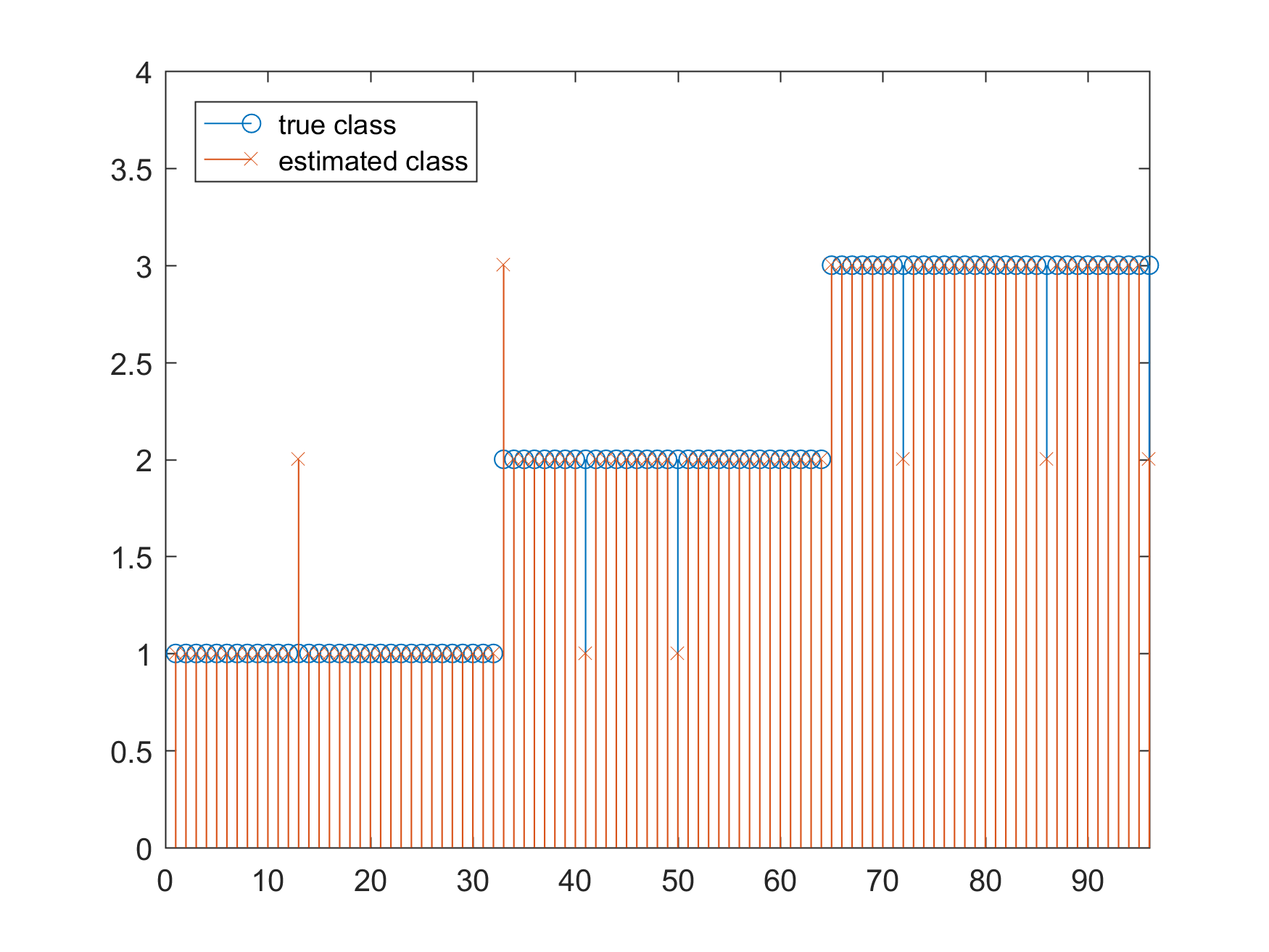}}
	\subfigure[]{\includegraphics[width=0.49\columnwidth]{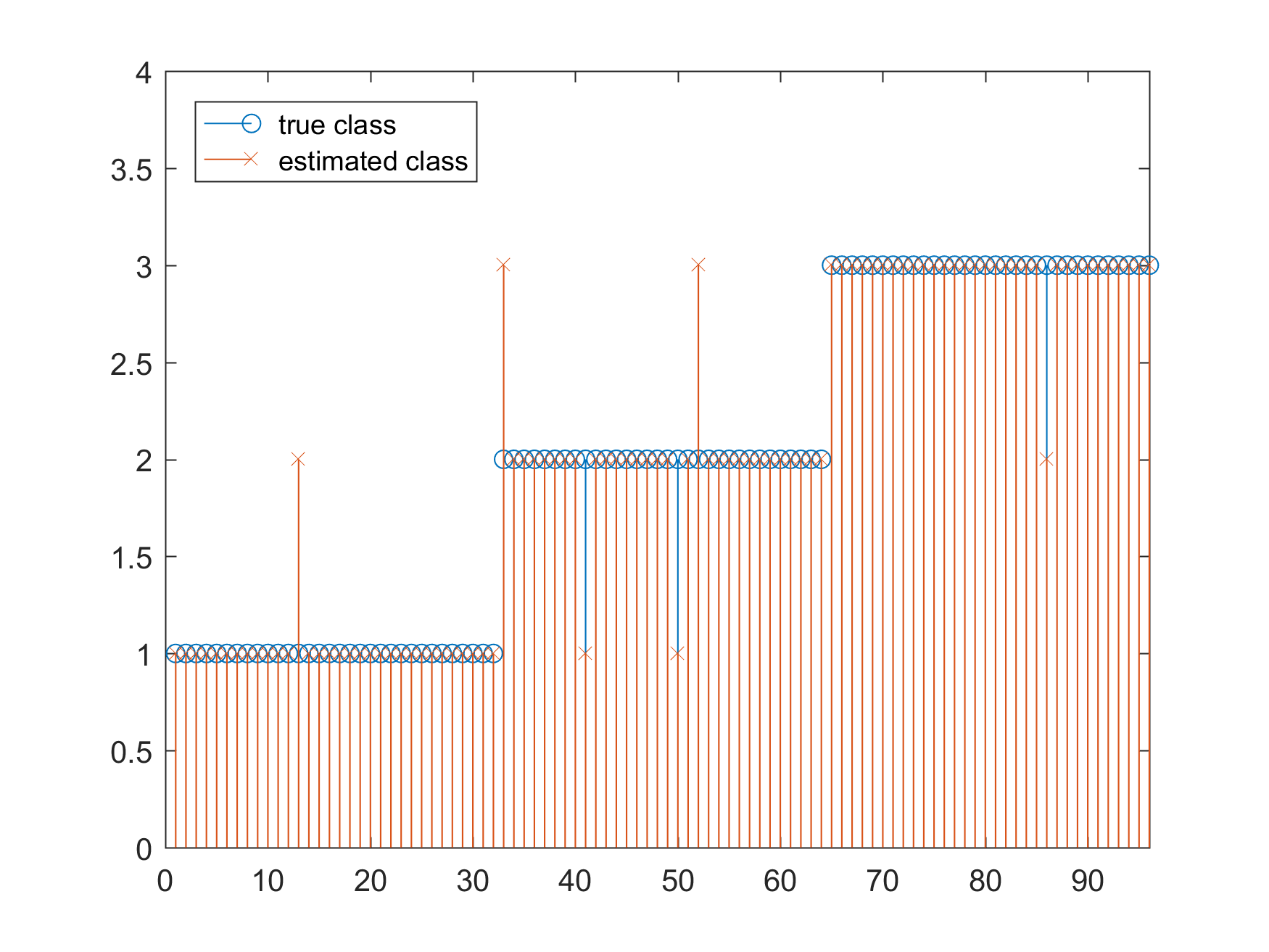}}
	\caption{Classification snapshots for $\sigma_{c,1}^2=20$ dB, $\sigma_{c,2}^2=30$ dB and $\sigma_{c,3}^2=40$ dB: (a) Proposition 1; (b) Proposition 2; (c) Proposition 3 with known $\bor$; (d) Proposition 3 with unknown $\bor$.}
	\label{figure5}
\end{figure}

\begin{figure}[htb] \centering
	\subfigure[]{\includegraphics[width=0.49\columnwidth]{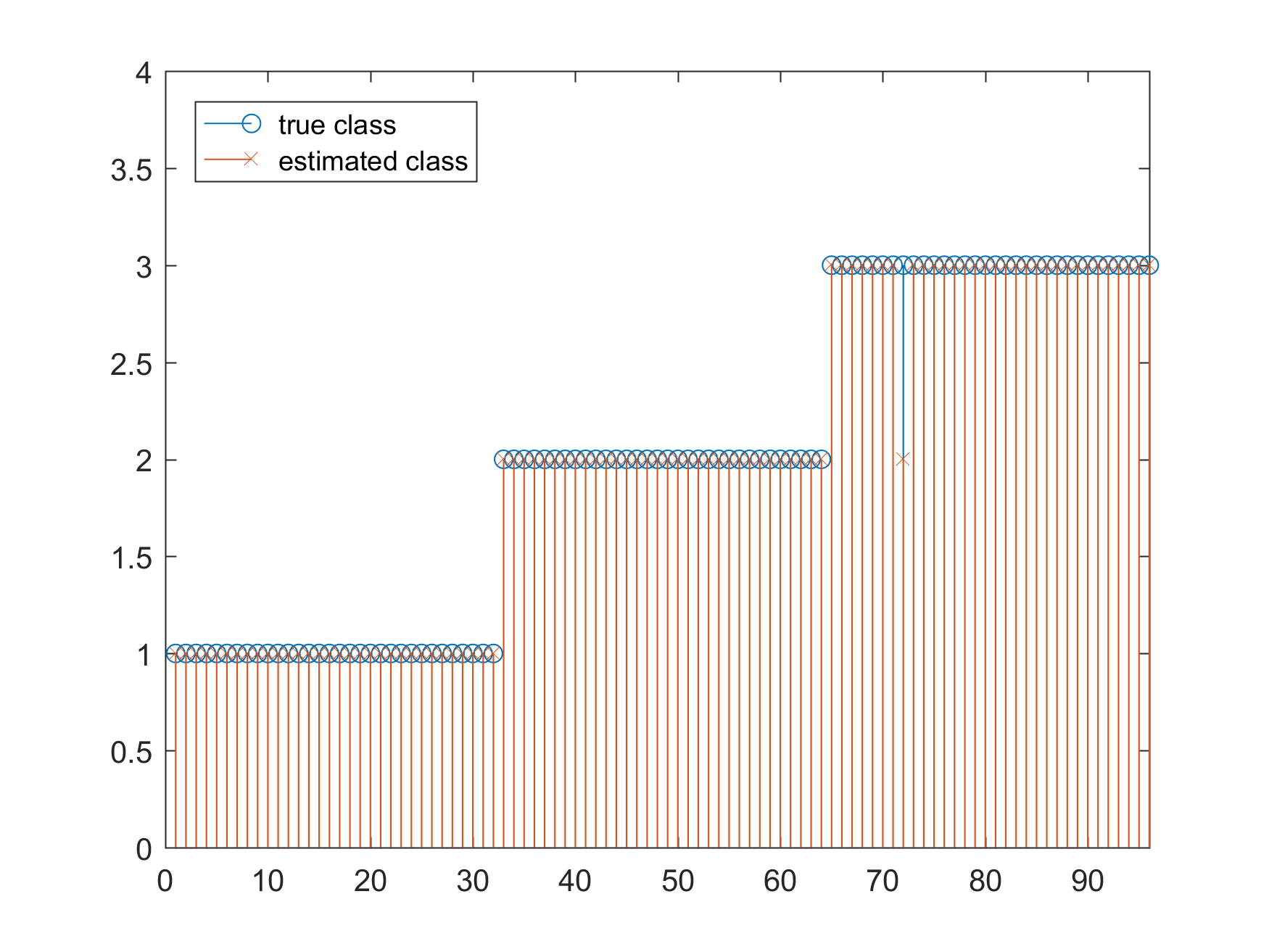}}
	\subfigure[]{\includegraphics[width=0.49\columnwidth]{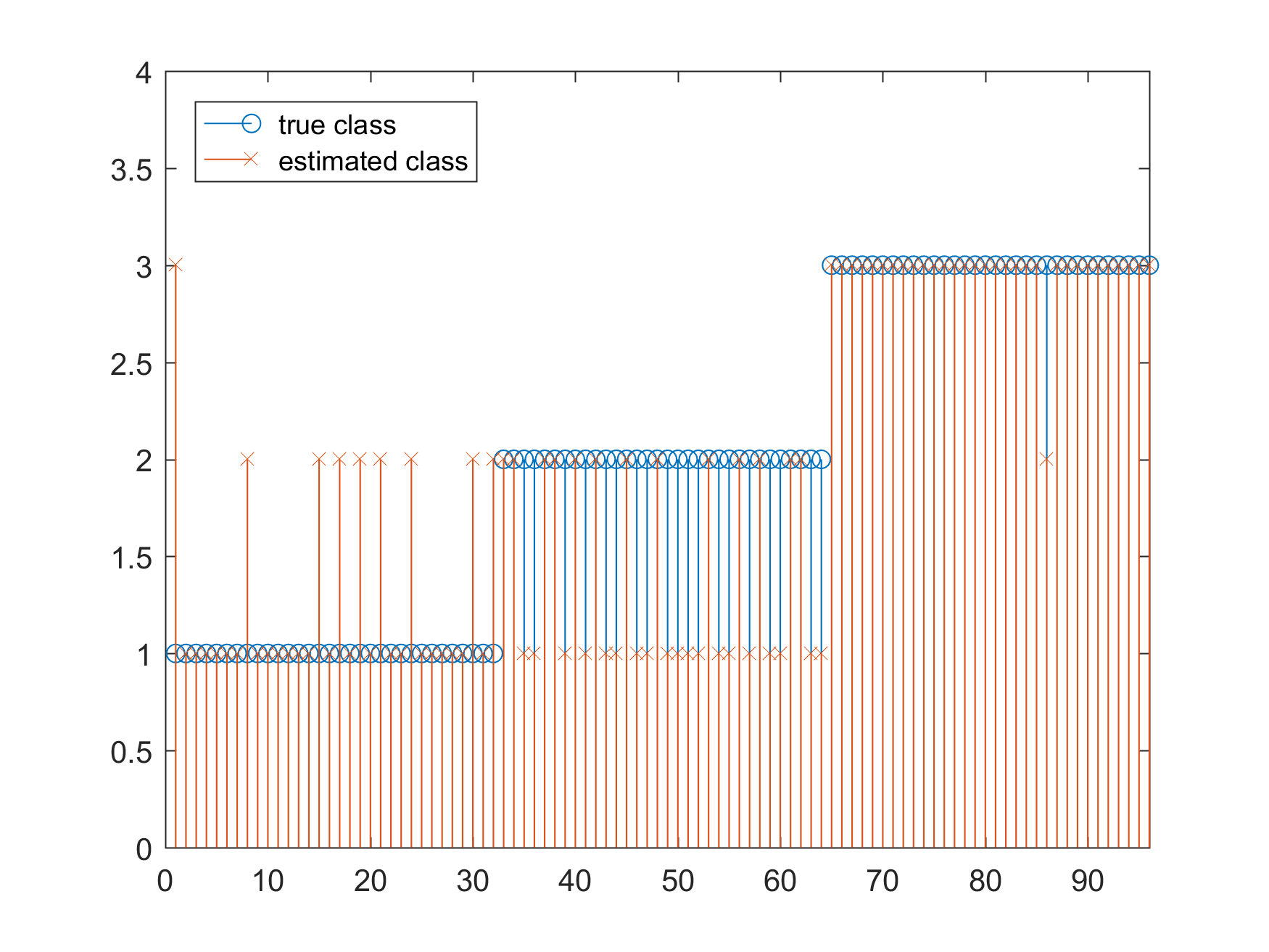}}
	\subfigure[]{\includegraphics[width=0.49\columnwidth]{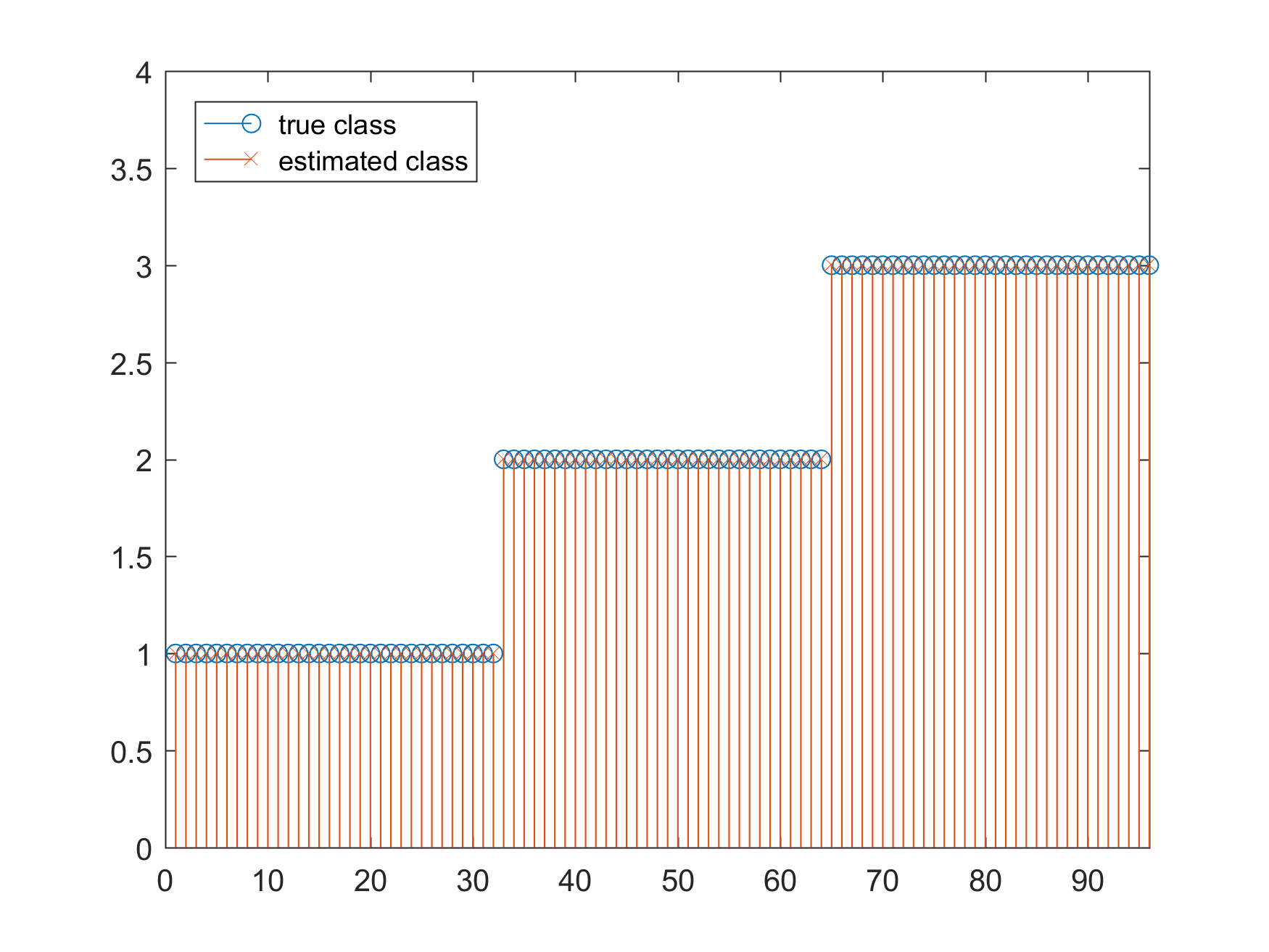}}
	\subfigure[]{\includegraphics[width=0.49\columnwidth]{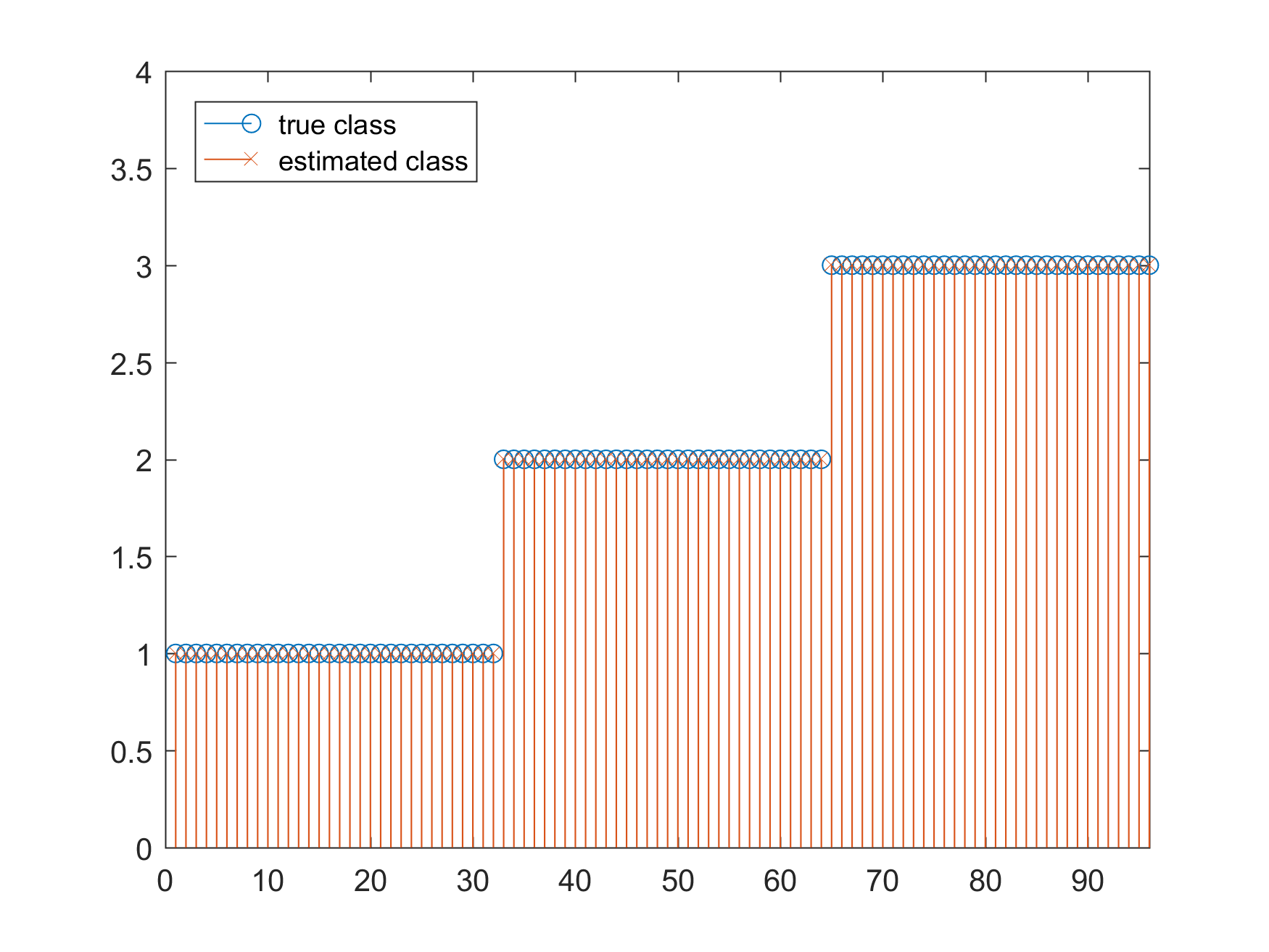}}
	\caption{Classification snapshots for $\sigma_{c,1}^2=20$ dB, $\sigma_{c,2}^2=35$ dB and $\sigma_{c,3}^2=50$ dB: (a) Proposition 1; (b) Proposition 2; (c) Proposition 3 with known $\bor$; (d) Proposition 3 with unknown $\bor$.}
	\label{figure6}
\end{figure}

\begin{table}[htbp]
	\centering
	\caption{RMSCE for different clutter powers and coavariance model \eqref{eq_model2}}
	\label{table3}
	\begin{spacing}{1.2}
		\begin{tabular}{|c|c|c|c|}
			\hline
			& case (1) & case (2) & case (3) \\
			\hline
			Proposition 1 & 26.29 & 5.87 & 1.03 \\
			\hline
			Proposition 2 & 43.89 & 33.08 & 29.21 \\
			\hline
			Proposition 3, $\bor$ known & 29.16 & 5.64 & 0.87 \\
			\hline
			Proposition 3, $\bor$ unknown & 29.94 & 5.81 & 0.97 \\
			\hline
		\end{tabular}
	\end{spacing}
\end{table}

A qualitative analysis is provided by Figures \ref{figure4}-\ref{figure6}, where Monte Carlo outcomes for 
the same clutter power configurations as in Section \ref{Subsec:case1} are shown (assuming $K_1=K_2=K_3=32$). 
The curves highlight that, for the considered parameters, the 
classification methods based on Propositions 1 and 3 share almost the same performance.
As to Proposition 2, it exhibits the worst performance. 

The RMSCE for different clutter power levels are shown in Table \ref{table3}. The obtained values confirm 
that the classification performances for Propositions 1 and 3 are very close to each other and are 
much better than that based on Proposition 2.

\begin{table*}[htbp]
	\centering
	\caption{RMSCE for different values of $K_l$ and coavariance model \eqref{eq_model2}}
	\label{table4}
	\begin{spacing}{1.2}{\small
		\begin{tabular}{|c|c|c|c|c|c|c|c|c|c|}
			\hline
			& [20,30,46] & [30,46,20] & [46,20,30] & [24,24,48] & [24,48,24] &[48,24,24] & [18,18,60] & [18,60,18] & [60,18,18]\\
			\hline
			Prop. 1 & 20.48 & 8.32 & 10.29 & 18.37 & 8.29 & 15.70 & 37.36 & 13.12 & 38.32\\
			\hline
			Prop. 2 & 37.66 & 29.19 & 24.99 & 35.84 & 27.90 & 25.97 & 44.59 & 22.11 & 35.72\\
			\hline
			Prop. 3, & 19.44 & 6.74 & 7.99 & 17.23 & 7.54 & 11.76 & 35.66 & 11.24 & 37.73  \\ 
			$\bor$ known &  &  &  &  &  &  &  & &\\
			\hline
			Prop. 3, & 18.55 & 6.83 & 7.29 & 16.29 & 7.58 & 11.00 & 35.28 & 10.84 & 38.28 \\ 
			$\bor$ unknown &  &  &  &  &  &  &  & &\\
			\hline
		\end{tabular}}
	\end{spacing}
\end{table*}

Table \ref{table4} contains the RMSCE values for different $K_l$ configurations. 
Inspection of the table indicates that the classification algorithm based on Proposition 3
can guarantee better performance than that 
based on Proposition 1. In addition, the knowledge of $\bor$ does not significantly affects the resulting
performance. Finally, Proposition 2 continues to return the highest error values.

\subsection{Real data}

In this section, we assess the performance analysis on real L-band land clutter data, recorded in 1985 using the MIT Lincoln Laboratory Phase One radar at the Katahdin Hill site, MIT Lincoln Laboratory. 
We consider datasets contained in the files $H067037.2$ and $H067038.3$, which 
are composed of $30720$ temporal returns from $76$ range cells with VV and HH-polarization, respectively.
More details about this dataset can be found in \cite{766939,937467,1413757} and references therein.


The 3-D clutter intensity field, from the Phase One file $H067037.2$, 
is plotted in Figure \ref{fig:matrixH067037_2}. It is evident the presence of two regions with different power levels (region 1 from cell 1 to cell 48 and region 2 from cell 49 to cell 76). 
This behavior, already observed in \cite{1413757}, is due to the fact that data were measured from range cells containing agricultural fields in contrast to windblown vegetation. Other five major terrain categories, distributed within the two major regions, are also evident, as indicated in figure.
%
The 3-D normalized intensity plot relative to the $H067038.3$ data file is reported in Figure \ref{fig:matrixH067038_3}).  Here, three major areas with different power levels can be identified.

These data are fed to the proposed algorithms and the used parameters are:
\begin{itemize}
\item $N=8$;
\item $K=75$;
\item $L=3$ or $5$;
\item a maximum number iterations of 10 (for both EM and alternating procedure).
\end{itemize}
Classification results, relative to the $H067037.2$ dataset, are reported in Figures \ref{fig:realdata_2001_1_L3} and \ref{fig:realdata_2001_1_L5}, for a number of classes of three and five, respectively. 
Data are characterized by small temporal variations of the power (variations in time on a given range cell, or on few cells) due to the inherent characteristic of the observed scene. Thus, the estimated classes are compared with power levels averaged over 100 temporal samples near the selected temporal $N$ samples.
The inspection of the figure points out that estimated classes follow the power profile 
for both $L=3$ and $L=5$. 
For this dataset, five classes allow to distinguish between all the five terrains indicated in Figure \ref{fig:matrixH067037_2}.

The classification results for dataset $H067038.3$ are shown in Figures \ref{fig:realdata_23001_86_L3} and \ref{fig:realdata_23001_86_L5}, respectively, and confirm what observed in the previous figures.

\begin{figure}
    \centering
    \includegraphics[scale=0.3]{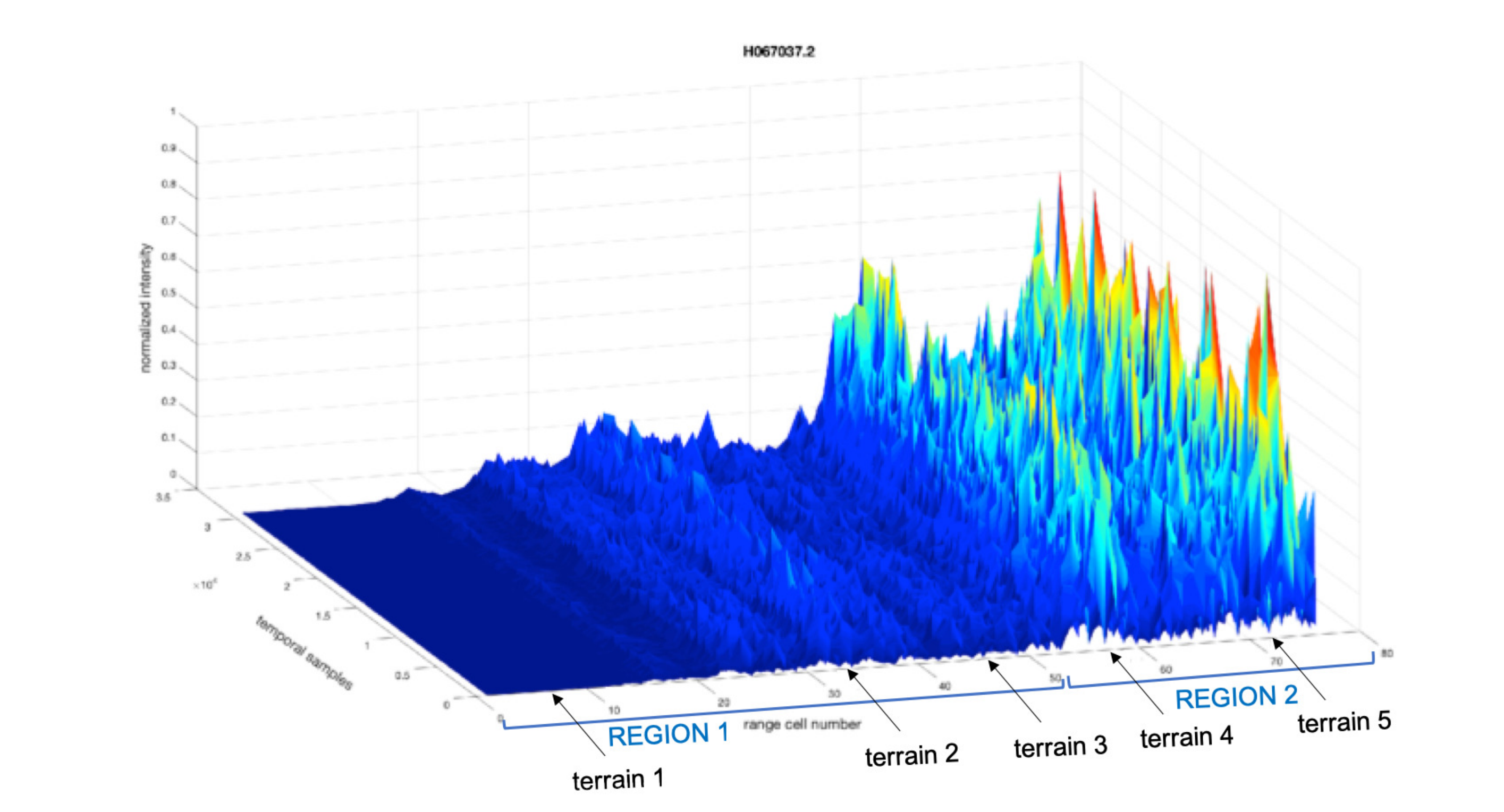}
       \caption{3-D normalized intensity field of clutter returns (H067037.2 dataset).}
    \label{fig:matrixH067037_2}
\end{figure}

\begin{figure}
    \centering
    \includegraphics[scale=0.17]{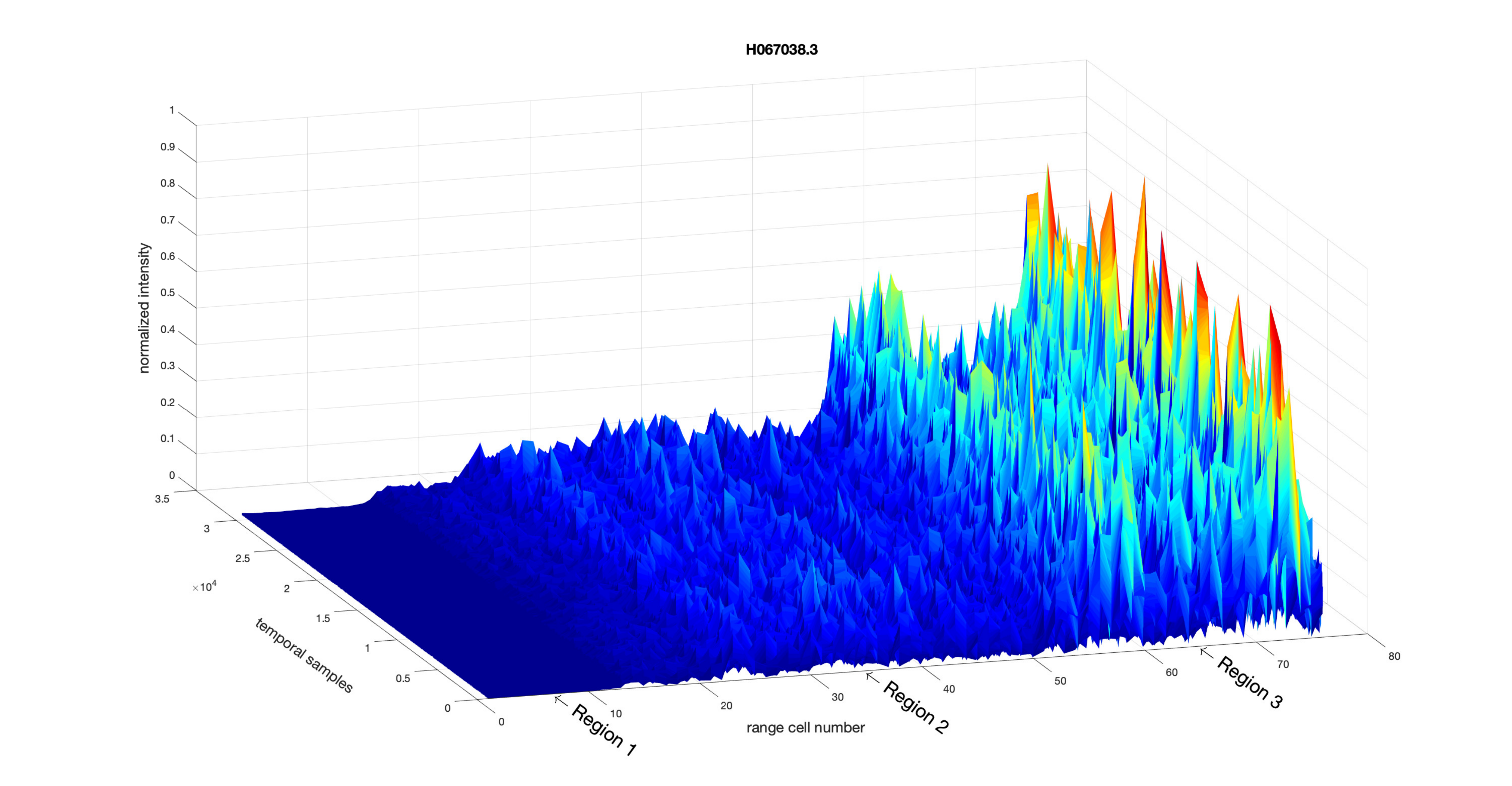}
       \caption{3-D normalized intensity field of clutter returns (H067038.3 dataset).}
    \label{fig:matrixH067038_3}
\end{figure}

\begin{figure}
    \centering
    \includegraphics[width=8.5cm,height=7cm]{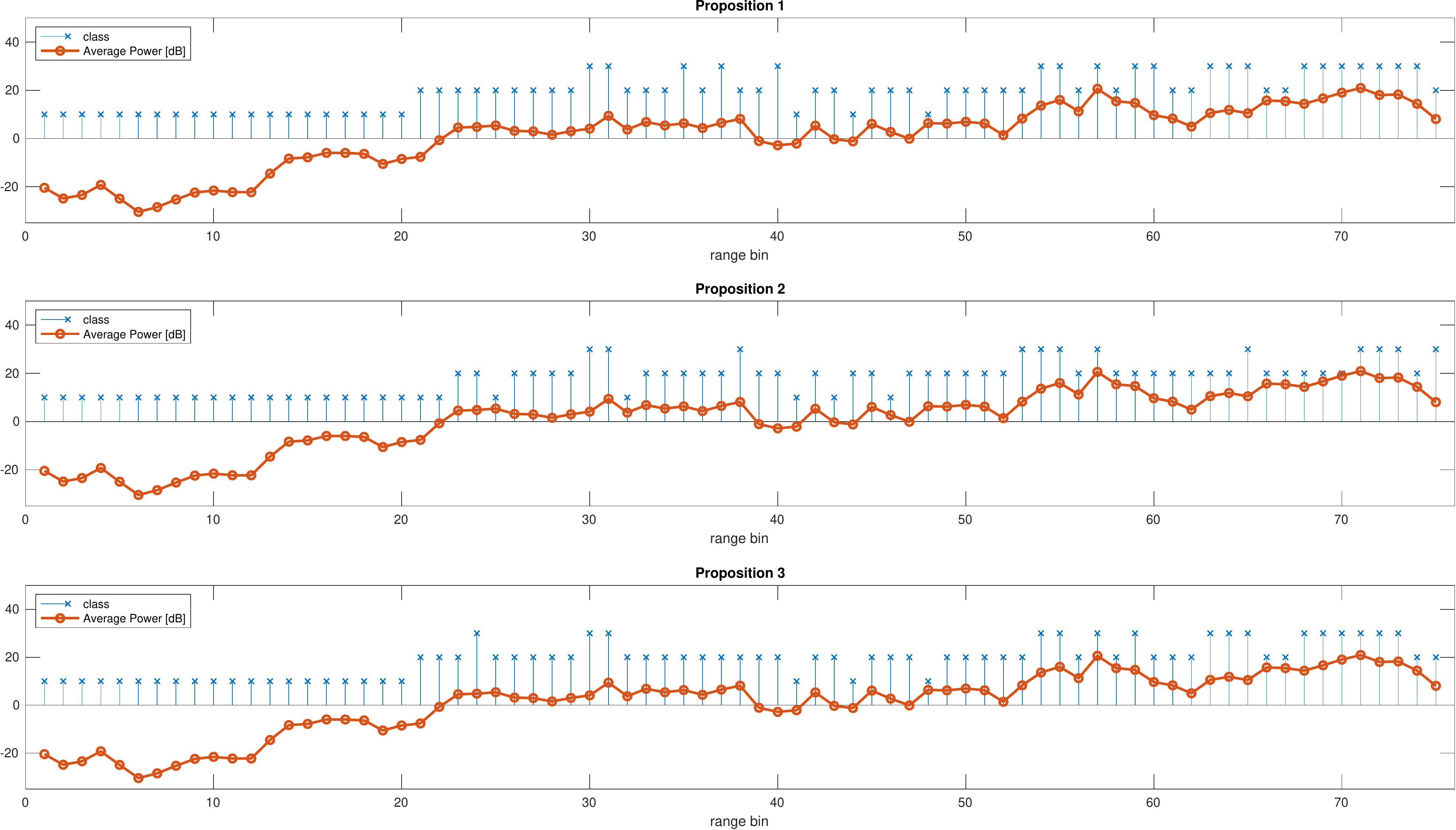}
       \caption{Average power and estimated classes for $L=3$ (H067037.2 dataset).}
    \label{fig:realdata_2001_1_L3}
\end{figure}

\begin{figure}
    \centering
    \includegraphics[width=8.5cm,height=7cm]{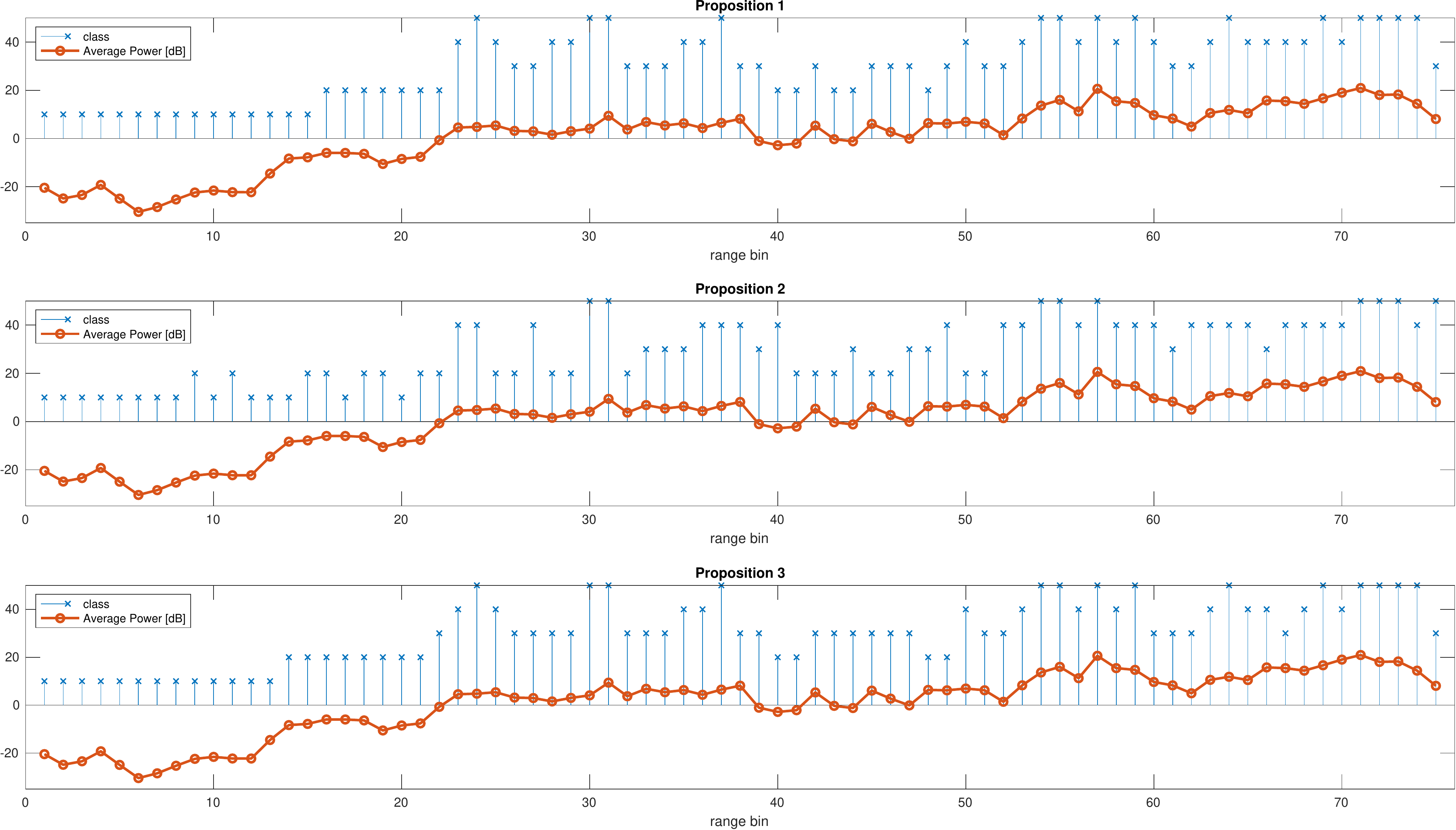}
       \caption{Average power and estimated classes for $L=5$ (H067037.2 dataset).}
    \label{fig:realdata_2001_1_L5}
\end{figure}

\begin{figure}
    \centering
    \includegraphics[width=8.5cm,height=7cm]{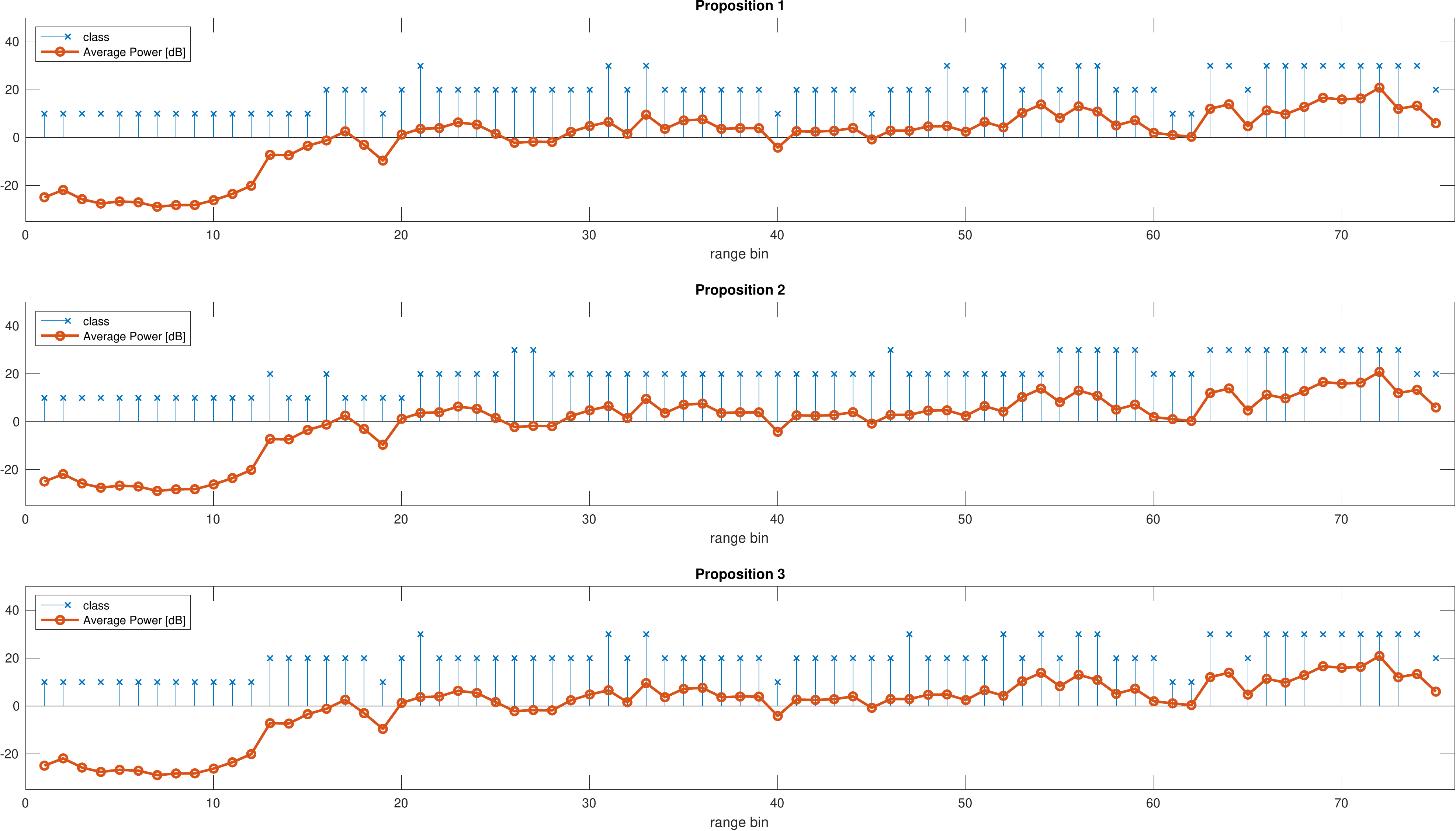}
       \caption{Average power and estimated classes for $L=3$ (H067038.3 dataset).}
    \label{fig:realdata_23001_86_L3}
\end{figure}

\begin{figure}
    \centering
    \includegraphics[width=8.5cm,height=7cm]{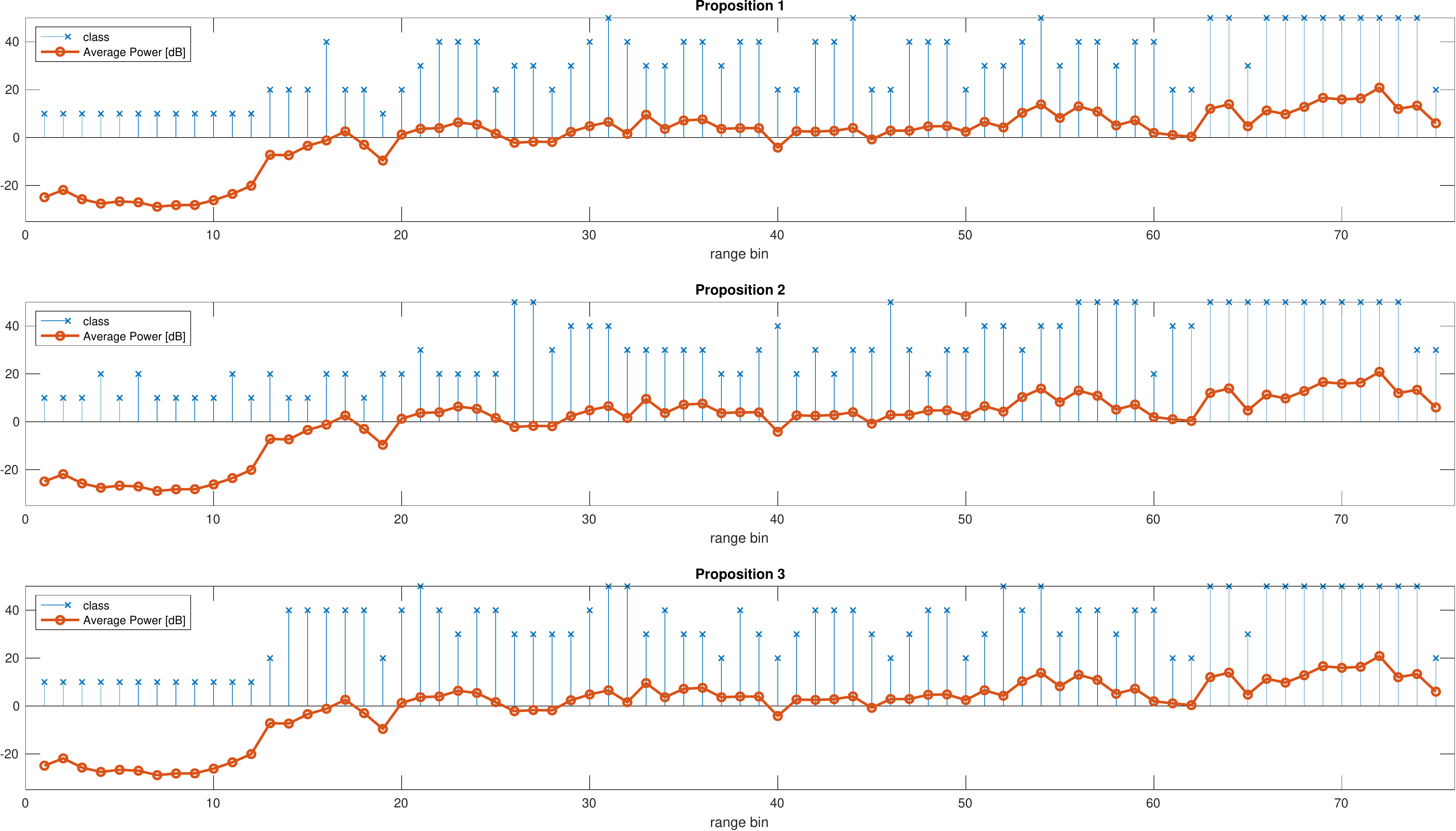}
       \caption{Average power and estimated classes for $L=5$ (H067038.3 dataset).}
    \label{fig:realdata_23001_86_L5}
\end{figure}

\section{Conclusions}
\label{Sec:Conclusions}
This paper has proposed several algorithms to classify clutter radar echoes with the goal
of partitioning the possibly heterogeneous training dataset into homogeneous subsets, which, then, can be used
for estimation/detection purposes. The algorithms have been designed using the EM algorithm in conjunction with the latent variable model. More precisely,
considering three different structures for the clutter covariance matrix (from the most
general case of a Hermitian structure to the specific one where diagonal loading is accounted for)
three different classification architectures have been introduced.
Performance analysis for both simulated and real data has clearly shown the capability of the proposed
approach to solve the problem of clutter data clustering. More importantly, these schemes can be used
as preliminary stage of a detection architecture, where the detection stage exploits the information
provided by the classifier to process homogeneous data.


Future research tracks include the design of clustering algorithms in the presence of outliers, which can be discarded once identified. Another issue is related to further structures for the clutter covariance matrix that can improve the estimation quality and, hence, detection performance of those receivers relying on such estimates. Finally, the 
design of architectures for the joint detection and classification of clutter edges represent an important 
extension of this work. All the above topics represent the current research activity.

\appendices

\section{Proof of Proposition \ref{Prop:EM-M_estimates01}}
\label{App:ProofProp1}
Let us consider the following problem 
\be
\widehat{\bsigma}^{(h)}=
\arg\max_{\bsigma}
g_1(\bM_1,\ldots,\bM_L),
\ee
which is tantamount to solving 
\be
\widehat{\bM}_l^{(h)}=
\arg\max_{\bM_l}
\underbrace{\sum_{k=1}^K q_k^{(h-1)}(l)
\left[
-\log\det(\bM_l)-\bz_k^\dag\bM_l^{-1}\bz_k
\right]}_{d({\boldmath M}_l)}
\ee
for each $l=1,\ldots,L$. To this end, we set to zero the first derivative of $d(\bM_l)$ with 
respect to $\bM_l$ \cite{hjorungnes2011complex}, namely
\begin{multline}
\frac{\partial}{\partial \bM_l}
[d(\bM_l)]= -(\bM_l^T)^{-1}\sum_{k=1}^K q_k^{(h-1)}(l)
\\
+(\bM_l^T)^{-1} \left[\sum_{k=1}^K q_k^{(h-1)}(l) \bz_k\bz_k^\dag\right]^T(\bM_l^T)^{-1}=\bzero.
\end{multline}
The solution of the above equation is given by
\be
\widehat{\bM}_l^{(h)}=
\frac{\ds \sum_{k=1}^K q_k^{(h-1)}(l) \bz_k\bz_k^\dag}
{\ds \sum_{k=1}^K q_k^{(h-1)}(l)},
\ee
and the proof is complete.

\section{Proof of Proposition \ref{Prop:EM-M_estimates02}}
\label{App:ProofProp2}
In order to come up with the estimates of the $\sigma^2_{c,l}$s and $\bM$,
we set to zero the first derivatives of $g_2(\bsigma^2_c,\bM)$ 
(with respect to  the $\sigma^2_{c,l}$s and $\bM$), namely
\begin{multline}
\forall l=1,\ldots,L: \ \frac{ \partial g_2(\bsigma^2_c,\bM)}{\partial \sigma^2_{c,l}} =
-\sum_{k=1}^K  q_k^{(h-1)}(l)
\\ 
\times
\left( \frac{N}{\sigma^2_{c,l}} - \frac{1}{\sigma^4_{c,l}}
\bz_k^{\dagger} \bM^{-1} \bz_k
\right)=0
\end{multline}
and 
\begin{multline}
\frac{ \partial g_2(\bsigma^2_c,\bM)}{\partial \bM} = -
\sum_{k=1}^K  \sum_{l=1}^L q_k^{(h-1)}(l)
\\ 
\times
\left( \bM^{-1} - \bM^{-1}  \frac{1}{\sigma^2_{c,l}}
\bz_k \bz_k^{\dagger} \bM^{-1} 
\right)^T=\bzero.
\end{multline}
The equations can be re-written as
\be
\sigma^2_{c,l} =
\frac{\sum_{k=1}^K  q_k^{(h-1)}(l)  \bz_k^{\dagger} \bM^{-1} \bz_k}{N \sum_{k=1}^K  q_k^{(h-1)}(l)}, 
\quad l=1, \ldots, L,
\label{eq1_EM}
\ee
and 
\be
\bM =\frac{1}{K}
{\sum_{k=1}^K \sum_{l=1}^L  q_k^{(h-1)}(l) 
\frac{\bz_k \bz_k^{\dagger}}{\sigma^2_{c,l}}},
\label{eq2_EM}
\ee
respectively,
where we have used the fact that
\be
\sum_{k=1}^K \sum_{l=1}^L  q_k^{(h-1)}(l)=K.
\ee
Since the equation system formed by \eqref{eq1_EM} and \eqref{eq2_EM} 
does not admit a closed-form solution,
we propose to resort to alternating maximization;
based on $(\hat{\sigma}^{2}_{c,l})^{{(h-1)}}$ and $\widehat{\bM}^{(h-1)}$
we first compute the $(\hat{\sigma}^{2}_{c,l})^{{(1),(h)}}$s
by plugging $\widehat{\bM}^{(h-1)}$
into eqs. (\ref{eq1_EM}); then, we compute $\widehat{\bM}^{(1),(h)}$ by plugging
the $(\hat{\sigma}^{2}_{c,l})^{{(1),(h)}}$s into eq. (\ref{eq2_EM}). 
This procedure can be iterated obtaining, after $t$ iterations, 
the $(\hat{\sigma}^2_{c,l})^{{(t),(h)}}$s and $\widehat{\bM}^{(t),(h)}$.
To conclude the proof we observe
that both EM and alternating maximization lead to a non decreasing sequence of likelihood values \cite{Recursive}.


\section{Proof of Proposition \ref{Prop:EM-M_estimates03}}
\label{App:ProofProp3}
First we re-write \eqref{eqn:objectiveFunctionProp3}
as follows
\begin{align*}
g_3(\sigma^2_n,\bR_1,\ldots,\bR_L) 
&= \sum_{k=1}^K\sum_{l=1}^L q_k^{(h-1)}(l)
\left[ -\log\det(\sigma^2_n\bI \right.
\\  &+ \left. \bR_l) 
-N\log\pi-\bz_k^\dag (\sigma^2_n+\bR_l)^{-1} \bz_k \right]
\end{align*}
and also as
\begin{align}
\nonumber
g^{\prime}_3(\sigma^2_n,\bR_1,\ldots,\bR_L) 
&=
\sum_{k=1}^K\sum_{l=1}^L q_k^{(h-1)}(l)
\left\{ -\log\det(\sigma^2_n\bI \right. 
\\ &+ \left. \bR_l) 
-\tr[(\sigma^2_n+\bR_l)^{-1} \bS_k] \right\}
\label{eqn:maxProb}
\end{align}
where $\bS_k=\bz_k\bz_k^\dag$. Now, let us consider the eigendecomposition of $\bR_l$, namely
$$
\bR_l = \bU_l \bLambda_l \bU_l^\dag
$$
where $\bU_l\in\C^{N\times N}$ 
is a unitary matrix whose columns are the eigenvectors of $\bR_l$ while
$\Lambda_l$ is the corresponding diagonal matrix
of the eigenvalues of $\bR_l$; $\bLambda_l$ can be represented as 
$\bLambda_l=\diag(\lambda_{l,1},\ldots,\lambda_{l,r_l},0,
\ldots,0)\in\R^{N\times N}$ with $\lambda_{l,1}\geq\ldots\geq \lambda_{l,r_l}>0$. 
It follows that the objective function becomes
\begin{align*}
\nonumber
&g^{\prime}_3(\sigma^2_n,\bR_1,\ldots,\bR_L) 
= \sum_{l=1}^L\sum_{k=1}^K  q_k^{(h-1)}(l)
\left\{ -\log\det(\sigma^2_n\bI \right.
\\ \nonumber &+ \left. \bLambda_l) -\tr[\bU_l (\sigma^2_n\bI+\bLambda_l)^{-1} \bU_l^\dag \bS_k]
\right\}
\\ \nonumber
&=\sum_{l=1}^L \Bigg\{  -\left( \sum_{k=1}^K  q_k^{(h-1)}(l)\right)\log \left[(\sigma^2_n)^{N-r_l}\prod_{m=1}^{r_l}
(\sigma^2_n+\lambda_{l,m})\right] \nonumber
\\
& -\tr\left[\bU_l (\sigma^2_n\bI+\bLambda_l)^{-1} \bU_l^\dag \bS_{l}^{(h-1)}\right]\Bigg\},
\end{align*}
where
$$
\bS_{l}^{(h-1)}=\sum_{k=1}^K  q_k^{(h-1)}(l)\bS_k.
$$
Replacing $\bS_{l}^{(h-1)}$ by its eigendecomposition, we also come up with
\begin{align*}
&\sum_{l=1}^L \Bigg\{  -\left( \sum_{k=1}^K  q_k^{(h-1)}(l)\right)\log \left[(\sigma^2_n)^{N-r_l}\prod_{m=1}^{r_l}
(\sigma^2_n+\lambda_{l,m})\right] \nonumber
\\
& -\tr\left[\bU_l (\sigma^2_n\bI+\bLambda_l)^{-1} \bU_l^\dag \bO_{l}^{(h-1)}\bGamma_{l}^{(h-1)}(\bO^{(h-1)}_{l})^\dag\right]\Bigg\}
\end{align*}
where $\bGamma_{l}^{(h-1)}=\diag(\gamma_{l,1}^{(h-1)},\ldots,\gamma_{l,N}^{(h-1)})$ with
$\gamma_{l,1}^{(h-1)}\geq\ldots\geq \gamma_{l,N}^{(h-1)}$ being the eigenvalues of $\bS_l^{(h-1)}$ and
$\bO^{(h-1)}_{l}$ the unitary matrix of the corresponding eigenvectors.
As a consequence, the objective function \eqref{eqn:maxProb} can also be recast as
\begin{multline*}
g^{\prime\prime}_3(\sigma^2_n,\bV_l, \bLambda_l, l=1,\ldots,L) =
\sum_{l=1}^L \Bigg\{  -q^{(h-1)}(l)
\\
\times\log \left[(\sigma^2_n)^{N-r_l}\prod_{m=1}^{r_l}
(\sigma^2_n+\lambda_{l,m})\right] 
\\
-\tr\left[\bV_l (\sigma^2_n\bI+\bLambda_l)^{-1} \bV_l^\dag\bGamma_{l}^{(h-1)}\right]\Bigg\}
\end{multline*}
where $q^{(h-1)}(l)=\sum_{k=1}^K  q_k^{(h-1)}(l)$ and $\bV_l=(\bO^{(h-1)}_{l})^\dag\bU_l$. Exploiting
{\em Theorem 1} of \cite{mirsky1959trace}, it is possible to show that $\forall l =1,\dots,L$
$$
\arg\max_{\bV_l} -\tr\Big[\bV_l (\sigma^2_n\bI+\bLambda_l)^{-1} 
\bV_l^\dag\bGamma_{l}^{(h-1)}\Big]= \bI,
$$
which implies that $\bU_l^{(h)}=\bO^{(h-1)}_{l}$. Then, we obtain
\begin{align}
\nonumber
&g^{\prime\prime\prime}_3(\sigma^2_n,\bLambda_l, l=1,\ldots,L) 
\\ 
\label{eqn:optProbSigma}
&=
\max_{\bV_l \atop l=1,\ldots,L} g^{\prime\prime}_3(\sigma^2_n,\bV_l, \bLambda_l, l=1,\ldots,L) 
\\ \nonumber &=
\sum_{l=1}^L \Bigg\{  -q^{(h-1)}(l)({N-r_l})\log \sigma^2_n   
-q^{(h-1)}(l)
\\ \nonumber &\times\sum_{m=1}^{r_l}
\log(\sigma^2_n+\lambda_{l,m})
-\sum_{m=1}^{r_l}\frac{\gamma_{l,m}^{(h-1)}}{\sigma^2_n+\lambda_{l,m}}
-\sum_{m=r_l+1}^{N}\frac{\gamma_{l,m}^{(h-1)}}{\sigma^2_n} \Bigg\}.
\end{align}
As the next step towards the final result, we set to zero the first derivative of the above objective function with respect to $\lambda_{l,m}$,
$m=1,\ldots,r_l$, namely
\begin{align}
&\frac{\partial}{\partial \lambda_{l,m}}\left[
-q^{(h-1)}(l)\log(\sigma^2_n+\lambda_{l,m})
-\frac{\gamma_{l,m}^{(h-1)}}{\sigma^2_n+\lambda_{l,m}}
\right]=0 \nonumber
\\
&\Rightarrow
-q^{(h-1)}(l)\frac{1}{(\sigma^2_n+\lambda_{l,m})}
+\frac{\gamma_{l,m}^{(h-1)}}{(\sigma^2_n+\lambda_{l,m})^2}=0 \nonumber
\\
&\Rightarrow
\hat{\lambda}_{l,m} = 
\left\{
\begin{array}{ll}
\frac{\gamma_{l,m}^{(h-1)}}{q^{(h-1)}(l)}-\sigma^2_n, &  \sigma^2_n < \frac{\gamma_{l,m}^{(h-1)}}{q^{(h-1)}(l)},
\\ 0, & \mbox{otherwise}.
\end{array}
\right.
\end{align}
After replacing ${\lambda}_{l,m}$ with $\hat{\lambda}_{l,m}$ in \eqref{eqn:optProbSigma}, 
the last optimization is
\begin{multline*}
\max_{\sigma^2_n}
\sum_{l=1}^L \Bigg\{  -q^{(h-1)}(l)({N-r_l})\log \sigma^2_n   
-q^{(h-1)}(l)\sum_{m=1}^{r_l}
\\
\log\left(\frac{\gamma_{l,m}^{(h-1)}}{q^{(h-1)}(l)}\right)
-r_l q^{(h-1)}(l)
-\sum_{m=r_l+1}^{N}\frac{\gamma_{l,m}^{(h-1)}}{\sigma^2_n} \Bigg\},
\end{multline*}
which can be solved by finding the zeros of the following function
\begin{align*}
&\frac{\partial}{\partial \sigma^2_n}\left[ 
\sum_{l=1}^L \Bigg\{  -q^{(h-1)}(l)({N-r_l})\log \sigma^2_n
-\sum_{m=r_l+1}^{N}\frac{\gamma_{l,m}^{(h-1)}}{\sigma^2_n}
\right] \nonumber
\\
&=-\frac{1}{\sigma^2_n}\sum_{l=1}^L q^{(h-1)}(l)({N-r_l})+\frac{1}{(\sigma^2_n)^2}
\sum_{l=1}^L \sum_{m=r_l+1}^{N}{\gamma_{l,m}^{(h-1)}}.
\end{align*}
The result is
\be
\hat{\sigma}^{2(h)}_n=\frac{\sum_{l=1}^L \sum_{m=r_l+1}^{N}{\gamma_{l,m}^{(h-1)}}}
{\sum_{l=1}^L q^{(h-1)}(l)({N-r_l})}.
\ee
Finally, the estimate of $\lambda_{l,m}$, $l=1, \ldots, L$, $m=1,\ldots,r_l$,  is given by
\be
\hat{\lambda}^{(h)}_{l,m}=
\left\{
\begin{array}{ll}
\frac{\gamma_{l,m}^{(h-1)}}{q^{(h-1)}(l)}-\hat{\sigma}^{2 (h)}_n, 
&  \hat{\sigma}^{2 (h)}_n < \frac{\gamma_{l,m}^{(h-1)}}{q^{(h-1)}(l)},
\\ 0, & \mbox{otherwise},
\end{array}
\right.
\ee
and the proof is complete.

%
\bibliographystyle{IEEEtran}
\bibliography{group_bib}
%
%
%
\end{document}